\pdfoutput=1
\documentclass[a4paper, USenglish, cleveref, autoref, thm-restate]{lipics-v2021}

\title{A Direct Translation from LTL with Past to Deterministic Rabin Automata}

\author{Shaun Azzopardi}{University of Malta, Malta}{shaun.azzopardi@gmail.com}{https://orcid.org/0000-0002-2165-3698}{}
\author{David Lidell}{University of Gothenburg, Sweden}{david.lidell@gu.se}{https://orcid.org/0000-0003-2760-339X}{}
\author{Nir Piterman}{University of Gothenburg, Sweden}{nir.piterman@gu.se}{https://orcid.org/0000-0002-8242-5357}{}

\authorrunning{S. Azzopardi, D. Lidell, N. Piterman}
\hideLIPIcs

\funding{This research is supported by the Swedish research council (VR) project (No. 2020-04963) and the ERC Consolidator grant DSynMA (No. 772459).}
\acknowledgements{We are grateful to Gerardo Schneider and the reviewers for their very useful feedback.}

\usepackage{mathtools}
\usepackage{hyperref}

\nolinenumbers

\newcommand{\ltlXsymb}{\mathbf{X}}
\newcommand{\ltlUsymb}{\mathbf{U}}
\newcommand{\ltlWsymb}{\mathbf{W}}
\newcommand{\ltlRsymb}{\mathbf{R}}
\newcommand{\ltlMsymb}{\mathbf{M}}
\newcommand{\ltlSsymb}{\mathbf{S}}
\newcommand{\ltlWSsymb}{\mathbf{\tilde{S}}}
\newcommand{\ltlNWSsymb}{\mathbf{B}}
\newcommand{\ltlNSsymb}{\mathbf{\tilde{B}}}
\newcommand{\ltlX}[0]{\ltlXsymb}
\newcommand{\ltlG}[0]{\mathbf{G}}
\newcommand{\ltlF}[0]{\mathbf{F}}
\newcommand{\ltlU}[0]{\, \ltlUsymb \,}
\newcommand{\ltlW}[0]{\, \ltlWsymb \,}
\newcommand{\ltlR}[0]{\, \ltlRsymb \,}
\newcommand{\ltlM}[0]{\, \ltlMsymb \,}
\newcommand{\ltlY}[0]{\mathbf{Y}}
\newcommand{\ltlWY}[0]{\mathbf{\tilde{Y}}}
\newcommand{\ltlO}[0]{\mathbf{O}}
\newcommand{\ltlH}[0]{\mathbf{H}}
\newcommand{\ltlS}[0]{\, \ltlSsymb \,}
\newcommand{\ltlWS}[0]{\, \ltlWSsymb \,}
\newcommand{\ltlNS}[0]{\, \ltlNSsymb \,}
\newcommand{\ltlNWS}[0]{\, \ltlNWSsymb \,}
\newcommand{\ltlUNOP}[0]{\mathit{op} \,}
\newcommand{\ltlBINOP}[0]{\, \mathit{op} \,}
\newcommand{\rw}[1]{\langle #1 \rangle}
\newcommand{\wc}[0]{\mathit{wc}}
\newcommand{\cwc}[0]{\mathit{rc}}
\newcommand{\sff}[0]{\mathit{sf}}
\newcommand{\psf}[0]{\mathit{psf}}
\newcommand{\af}[0]{\mathit{af}}
\newcommand{\afc}[0]{\mathit{af_{\ell}}}
\newcommand{\puc}[0]{\mathit{pu_{\ell}}}
\newcommand{\G}[3]{\mathcal{G}_{#2, #3}^{#1}}
\newcommand{\F}[3]{\mathcal{F}_{#2, #3}^{#1}}
\newcommand{\GF}[3]{\mathcal{GF}_{#2, #3}^{#1}}
\newcommand{\FG}[3]{\mathcal{FG}_{#2, #3}^{#1}}
\newcommand{\Var}{\mathit{Var}}
\newcommand{\placeholder}{\makebox[1ex]{\textbf{$\cdot$}}}

\newcommand{\entcon}[2]{C_{#1, #2}}
\newcommand{\entseq}[2]{\vec{C}_{#1, #2}}
\newcommand{\compseq}[2]{\circ \, \entseq{#1}{#2}}
\newcommand{\entconw}[3]{C^{#3}_{#1, #2}}


\begin{document}

\maketitle

\begin{abstract}
We present a translation from linear temporal logic \emph{with past} to
deterministic Rabin automata. The translation is direct in the sense
that it does not rely on intermediate non-deterministic automata,
and asymptotically optimal, resulting in Rabin automata of doubly
exponential size. It is based on two main notions. One is that it
is possible to encode the history contained in the prefix of a word, as relevant for the formula under consideration, by performing simple rewrites of the formula itself. As a consequence, a formula involving past operators can (through such rewrites, which involve alternating between weak and strong versions of past operators in the formula’s syntax tree) be correctly evaluated at an arbitrary point in the future without requiring backtracking through the word. The other is that this allows us to generalize to linear temporal logic with past the result that the language of a pure-future formula can be decomposed into a Boolean combination of simpler languages, for which deterministic automata with simple acceptance conditions are easily constructed.
\end{abstract}

\section{Introduction}
Finite-state automata over infinite words, commonly referred to as $\omega$-automata, have been studied as models of computation since their introduction in the 1960s. Their introduction was followed by extensive investigations into their expressive power, closure properties, and the decidability and complexity of related decision problems, such as non-emptiness and language containment. In particular, it was soon established that determinization constructions of such automata are considerably more complex than they are for their finite word counterparts, for which a simple subset construction is sufficient.

In the 1970s, Pnueli~\cite{pnueli_ltl} proposed linear temporal logic ($\mathit{LTL})$ as a specification language for the analysis and verification of programs, based on the idea that a program can be viewed in the context of a stream of interactions between it and its environment. It has since become ubiquitous in both academia and the industry due to the perceived balance of its expressive power and the computational complexity of its related decision procedures. The intimate semantic connection between $\mathit{LTL}$ and $\omega$-automata further motivated research into the properties of both, also from a practical point of view. In addition to being of theoretical interest, the fact that multiple automata-based applications of $\mathit{LTL}$ as a specification language \textendash{} such as reactive synthesis and probabilistic model checking \textendash{} require deterministic automata, has drawn additional attention to the study of efficient translation procedures from $\mathit{LTL}$ to deterministic $\omega$-automata.

In the classic approach to determinization, the given $\mathit{LTL}$ formula is first translated to a non-deterministic Büchi automaton. This automaton is subsequently determinized, for example using Safra's procedure~\cite{safra_determinization}, or the more modern Safra/Piterman variant~\cite{piterman_determinization}. Such constructions are both conceptually complex and difficult to implement. Moreover, information about the structure of the given formula is lost in the initial translation to Büchi automata; in particular, because of the generality of such determinization procedures, they cannot take advantage of the fact that $\mathit{LTL}$ is less expressive than $\omega$-automata.

In 2020, Esparza et al.~\cite{esparza_unified_translation} presented a novel translation from $\mathit{LTL}$ to various automata, that is asymptotically optimal in both the deterministic and non-deterministic cases. The translation is direct in the sense that it avoids the intermediate steps of the classic approach, which involve employing a variety of separate translation procedures. In particular, for deterministic automata, it forgoes Safra-based constructions. Instead, the language of the formula under consideration is decomposed into a Boolean combination of simpler languages, for which deterministic automata with simple acceptance conditions can easily be constructed using what the authors have dubbed the ``after-function''; that such a decomposition exists is a fundamental result named the ``Master Theorem''. These simpler automata are then combined into the desired final automaton using basic product or union operations according to the structure of the decomposition.

In this paper, we consider past linear temporal logic ($\mathit{pLTL}$); the extension of $\mathit{LTL}$ that includes the past operators ``Yesterday'' and ``Since'', analogous to the standard operators ``Next'' and ``Until'', respectively. We adapt the Master Theorem and generalize the derived $\mathit{LTL}$-to-deterministic-Rabin-automata translation to $\mathit{pLTL}$, while maintaining its optimal asymptotic complexity. The merits of $\mathit{LTL}$ extended with past operators were argued by Lichtenstein et al.~\cite{lichtenstein_glory}, who also showed that their addition does not result in a more expressive logic with respect to initial equivalence of formulae. At the same time, the complexity of satisfiability/validity- and model checking remains PSPACE-complete for $\mathit{pLTL}$~\cite{sistla_complexity_pltl}. When it comes to determinization, as Esparza's approach \cite{esparza_unified_translation} applies only to (future) $\mathit{LTL}$, the only option is the two-step approach: translate $\mathit{pLTL}$ to nondeterministic B\"uchi automata \cite{piterman_handbook} and convert to deterministic parity/Rabin automata \cite{safra_determinization,piterman_determinization}. 
Our generalization to $\mathit{pLTL}$ is of both theoretical and practical interest, for two main reasons. First, certain properties are more naturally and elegantly expressed with the help of past operators. Secondly, there exist formulae in $\mathit{pLTL}$ such that all (initially) equivalent $\mathit{LTL}$ formulae are exponentially larger~\cite{markey_succinct}. Both of these properties can be exemplified by considering the natural language specification \textit{``At any point in time, $p$ should occur if and only if $q$ and $r$ have occurred at least once in the past''}. Expressing this in $\mathit{LTL}$ requires explicitly describing the possible desired orders of occurrences of $p$, $q$, and $r$:
\begin{equation*}
\begin{gathered}    
\big( (\neg p \land \neg q) \ltlW (r \land ((\neg p \land \neg q) \ltlW (p \land q))) \lor
(\neg p \land \neg r) \ltlW (q \land ((\neg p \land \neg r) \ltlW (p \land r))) \big) \\ 
\land \, \ltlG (p \Rightarrow \ltlX \ltlG p).
\end{gathered}
\end{equation*}
The same specification is very intuitively and succinctly expressed in $\mathit{pLTL}$ by the formula $\ltlG (p \Leftrightarrow \ltlO q \land \ltlO r)$.

The main contributions of this paper are an adaptation of the Master Theorem for $\mathit{pLTL}$ and a utilization of this in the form of an asymptotically optimal direct translation from $\mathit{pLTL}$ to deterministic Rabin automata. The paper is structured in the following manner: In Section~\ref{section_preliminaries}, we define the syntax and semantics of linear temporal logic with past, infinite words and $\omega$-automata, and the notion of propositional equivalence. As the automata we aim to construct are one-way,
we need a way to encode information about the input history directly into the formula being translated; the machinery required to accomplish this is introduced in Section~\ref{section_encodingthepast}. The foundation of the translation from $\mathit{pLTL}$ to Rabin automata is the after-function of Section~\ref{section_afterfunction}. The subsequent Sections~\ref{section_stability} and~\ref{section_toautomata} are adaptations of the corresponding sections in Esparza et al.~\cite{esparza_unified_translation}. The decomposition of the language of the formula to be translated into a Boolean combination of simpler languages requires considering the limit-behavior of the formula. This notion is made precise in Section~\ref{section_stability}, which finishes with a presentation of the Master Theorem for $\mathit{pLTL}$. Section~\ref{section_toautomata} describes how to create deterministic automata from the simpler languages of the decomposition, and how to combine them into a deterministic Rabin automaton. 
Finally, we conclude with a brief discussion in Section~\ref{section_discussion}.

\section{Preliminaries}
\label{section_preliminaries}
\subsection{Infinite Words and \texorpdfstring{$\omega$}{Omega}-automata}
An infinite word $w$ over a non-empty finite alphabet $\Sigma$ is an infinite sequence $\sigma_0, \sigma_1, \dots$ of letters from $\Sigma$. Given an infinite word $w$, we denote the finite infix $\sigma_t, \sigma_{t+1}, \dots, \sigma_{t + s - 1}$ of $w$ by $w_{ts}$. If $t = s$, then $w_{ts}$ is defined as representing the empty word $\epsilon$. 
Note that no ambiguity will arise as we use parentheses whenever required; for example, $w_{(st)(en)}$ rather than $w_{sten}$
We denote the infinite suffix $\sigma_t, \sigma_{t+1}, \dots$ of $w$ by $w_t$. We will also consider finite words, using the same infix- and suffix notation.

An \textit{$\omega$-automaton} over an alphabet $\Sigma$ is a quadruple $(Q, Q_0, \delta, \alpha)$, where $Q$ is a finite set of states, $Q_0 \subseteq Q$ a non-empty set of initial states, $\delta \in Q \times \Sigma \to 2^Q$ a (partial) transition function, and $\alpha$ a set constituting its acceptance condition. In the case where $\vert Q_0 \vert = 1$ and $\vert \delta(q, \sigma) \vert \leq 1$ for all $q \in Q$ and all $\sigma \in \Sigma$, the automaton is called \textit{deterministic}.  
For deterministic automata we write $\delta:Q \times \Sigma \rightarrow Q$.
Given an $\omega$-automaton $\mathcal{A} = (Q, Q_0, \delta, \alpha)$ and an infinite word $w = \sigma_0, \sigma_1, \dots$, both over the same alphabet, a \textit{run} of $\mathcal{A}$ on $w$ is a sequence of states $r = r_0, r_1, \dots$ of $Q$ such that $r_0 \in Q_0$ and $r_{i + 1} \in \delta(r_i, \sigma_i)$ for all $i \geq 0$. Given such a run $r$, we write $\mathit{Inf}(r)$ to denote the set of states appearing infinitely often in $r$.

In this paper, we consider three particular classes of $\omega$-automata, which differ only by their acceptance condition $\alpha$. \textit{Büchi}- and \textit{co-Büchi} automata are $\omega$-automata with $\alpha$ as a set $Q' \subseteq Q$. A Büchi automaton \textit{accepts} the infinite word $w$ iff there exists a run $r$ on $w$ such that $\mathit{Inf}(r) \cap Q' \neq \varnothing$, while a co-Büchi automaton accepts $w$ iff there exists a run $r$ on $w$ such that $\mathit{Inf}(r) \cap Q' = \varnothing$. We will also consider \textit{Rabin automata}. A Rabin automaton has a set of subsets $R \subseteq 2^Q \times 2^Q$ as acceptance condition, and accepts $w$ iff there exists a run $r$ on $w$ and a pair $(A, B) \in R$ such that $\mathit{Inf}(r) \cap A = \varnothing$ and $\mathit{Inf}(r) \cap B \neq \varnothing$.
We write $DBA$, $DCA$, and $DRA$ to refer to deterministic Büchi-, co-Büchi-, and Rabin automata, respectively.

We use a specialized form of cascade composition of two automata: \href{https://en.wikipedia.org/wiki/Millstone}{bed automaton and runner automaton}.
Bed automata are automata without an acceptance condition.
Runner automata read input letters and (next) states of bed automata: given a bed automaton $\mathcal{A}$ with $S$ as set of states, a runner automaton is $\mathcal{B} = (S,Q,Q_0,\delta,\alpha)$, where $Q$, $Q_0$, and $\alpha$ are as before and $\delta:Q\times S \times \Sigma \rightarrow 2^Q$ is the transition function.
The composition $\mathcal{A}{\ltimes}\mathcal{B}$ is the automaton $(Q \times S, Q_0 \times S_0, \overline{\delta}, \overline{\alpha})$, where $\overline{\delta}(q,s,\sigma)=\{(q', s') ~|~ q' \in \delta_B(q,s',\sigma), s' \in \delta_\mathcal{A}(s,\sigma) \}$ and either $\overline{\alpha} = \alpha \times S$ if $\mathcal{B}$ is a Büchi- or co-Büchi automaton, or $\overline{\alpha} = \bigcup \{ (A \times S) \times (B \times S) \mid (A, B) \in \alpha \}$ if $\mathcal{B}$ is a Rabin automaton.
Note that $\mathcal{A}{\ltimes}\mathcal{B}$ is an $\omega$-automaton with $\mathcal{B}$'s acceptance, e.g., if $\mathcal{A}$ is deterministic and $\mathcal{B}$ is a deterministic co-B\"uchi runner automaton, then $\mathcal{A}{\ltimes} \mathcal{B}$ is a DCA.

\subsection{Linear Temporal Logic with Past}
\label{subsection_pltl}
Given a non-empty finite set of propositional variables $AP$, the well-formed formulae $\varphi$ of $\mathit{pLTL}$ are generated by the following grammar:
\begin{equation*}
    \varphi \Coloneqq \top \mid \bot \mid p \mid \neg \varphi \mid \varphi \land \varphi \mid \varphi \lor \varphi \mid \ltlX \varphi \mid \varphi \ltlU \varphi \mid \ltlY \varphi \mid \varphi \ltlS \varphi,
\end{equation*}
where $p \in AP$. Given a formula $\varphi$, we write $\Var(\varphi)$ to denote the set of all atomic propositions appearing in $\varphi$. Given a formula $\varphi$, natural number $t$, and infinite word $w = \sigma_0, \sigma_1, \dots$ over $2^{\Var(\varphi)}$, we write $(w, t) \models \varphi$ to denote that $w$ satisfies $\varphi$ at index $t$. The meaning of this is made precise by the following inductive definition:
\begin{equation*}
\begin{aligned}
(w, t) &\models \top &
(w, t) &\not\models \bot \\
(w, t) &\models p \mbox{ iff } p \in \sigma_t &
(w, t) &\models \neg \varphi \mbox{ iff } (w, t) \not\models \varphi \\
(w, t) &\models \varphi \land \psi \mbox{ iff } (w, t) \models \varphi \mbox{ and } (w, t) \models \psi &
(w, t) &\models \varphi \lor \psi \mbox{ iff } (w, t) \models \varphi \mbox{ or } (w, t) \models \psi \\
(w, t) &\models \ltlX \varphi \mbox{ iff } (w, t + 1) \models \varphi &
(w, t) &\models \ltlY \varphi \mbox{ iff } t > 0 \mbox{ and } (w, t - 1) \models \varphi \\
(w, t) &\models \varphi \ltlU \psi \mbox{ iff } \exists r \geq t \, . \, \mathrlap{\left( (w, r) \models \psi \mbox{ and } \forall s \in [t, r) \, . \, (w, s) \models \varphi \right)} \\
(w, t) &\models \varphi \ltlS \psi \mbox{ iff } \exists r \leq t \, . \, \mathrlap{\left( (w, r) \models \psi \mbox{ and } \forall s \in (r, t] \, . \, (w, s) \models \varphi \right).}
\end{aligned}
\end{equation*}
When $t = 0$ we omit the index and simply write $w \models \varphi$. Observe that $\ltlY$ is almost exactly the past analogy of $\ltlX$, with a similar relationship between $\ltlSsymb$ and $\ltlUsymb$. They differ in that the past is bounded; in particular, a formula $\ltlY \varphi$ is never satisfied at $t = 0$.
The \textit{language} of a formula $\varphi$, denoted $\mathcal{L}(\varphi)$, is the set of all infinite words $w$ such that $w \models \varphi$. Two $\mathit{pLTL}$ formulae $\varphi$ and $\psi$ are semantically equivalent, denoted $\varphi \equiv \psi$, iff $\mathcal{L}(\varphi) = \mathcal{L}(\psi)$.

In addition to the above, we will also consider the following derived operators:
\begin{equation*}
\begin{aligned}
    \ltlF \varphi &\coloneqq \top \ltlU \varphi & \ltlO \varphi &\coloneqq \top \ltlS \varphi & \ltlG \varphi &\coloneqq \neg \ltlF \neg \varphi \\
    \ltlH \varphi &\coloneqq \neg \ltlO \neg \varphi &
    \varphi \ltlW \psi &\coloneqq \varphi \ltlU \psi \lor \ltlG \varphi & \varphi \ltlWS \psi &\coloneqq \varphi \ltlS \psi \lor \ltlH \varphi \\
    \varphi \ltlM \psi &\coloneqq \psi \ltlU (\varphi \land \psi) & \varphi \ltlNWS \psi &\coloneqq \psi \ltlS (\varphi \land \psi) &
    \varphi \ltlR \psi &\coloneqq \psi \ltlW (\varphi \land \psi) \\ \varphi \ltlNS \psi &\coloneqq \psi \ltlWS (\varphi \land \psi) &
    \ltlWY \varphi &\coloneqq \ltlY \varphi \lor \neg \ltlY \top.
\end{aligned}
\end{equation*}
The past operators above are defined in analogy with their standard future counterparts. An important exception is the \textit{weak yesterday} operator $\ltlWY$, which is similar to $\ltlY$. However, a formula $\ltlWY \varphi$ is always satisfied at $t = 0$.

A $\mathit{pLTL}$ formula that is neither atomic nor whose syntax tree is rooted with a Boolean operator is called a \textit{temporal} formula, and a $\mathit{pLTL}$ formula whose syntax tree is rooted with an element of $\{ \ltlY, \ltlWY, \ltlSsymb, \ltlWSsymb, \ltlNWSsymb, \ltlNSsymb
\}$ a \textit{past} formula. Finally, a $\mathit{pLTL}$ formula that is a proposition or the negation thereof is \textit{propositional}. We write $\sff(\varphi)$ to denote the set of propositional and temporal subformulae of $\varphi$, and $\psf(\varphi)$ the set of past subformulae of $\varphi$. The \textit{size} of a $\mathit{pLTL}$ formula $\varphi$, denoted $\vert \varphi \vert$, is defined as the number of nodes of its syntax tree that are either temporal or propositional.

A $\mathit{pLTL}$ formula where negations only appear before atomic propositions is in \textit{negation normal form}. Observe that with the derived operators above, an arbitrary $\mathit{pLTL}$ formula can be rewritten in negation normal form with a linear increase in size. For the remainder of the paper, when we write ``formula'' we implicitly refer to $\mathit{pLTL}$ formulae in negation normal form, with no occurrences of $\ltlF, \ltlG, \ltlO$ or $\ltlH$. Subformulae rooted with either of these four operators can be replaced by equivalent formulae of the same size: every subformula of the form $\ltlF \psi$ can be replaced with $\top \ltlU \psi$ and every subformula of the form $\ltlG \psi$ with $\psi \ltlW \bot$, and analogously for $\ltlO \psi$ and $\ltlH \psi$. While these four derived operators are not part of the syntax under consideration \textendash{} for the purpose of keeping the presentation more concise \textendash{} we will occasionally use them as convenient shorthand. When we write ``word'' we implicitly refer to infinite words, unless otherwise stated.

We conclude this section by defining the notion of \textit{propositional equivalence}, which is a stronger notion of equivalence than that given by the semantics of \textit{pLTL}. It is relatively simple to determine whether two formulae are propositionally equivalent, which makes the notion useful in defining the state spaces of automata as equivalence classes of formulae. We also state a lemma that allows us to lift functions defined on formulae to the propositional equivalence classes they belong to. Both are due to Esparza et al.~\cite{esparza_unified_translation}.

\begin{definition}[Propositional Semantics of $\mathit{pLTL}$]
Let $\mathcal{I}$ be a set of formulae and $\varphi$ a formula. The propositional satisfaction relation $\mathcal{I} \models_p \varphi$ is inductively defined as
\begin{equation*}
\begin{aligned}
    \mathcal{I} &\models_p \top \quad &
    \mathcal{I} &\models_p \psi \land \xi \mbox{ iff } \mathcal{I} \models_p \psi \mbox{ and } \mathcal{I} \models_p \xi \\
    \mathcal{I} &\not\models_p \bot &
    \mathcal{I} &\models_p \psi \lor \xi \mbox{ iff } \mathcal{I} \models_p \psi \mbox{ or } \mathcal{I} \models_p \xi,
\end{aligned}
\end{equation*}
with $\mathcal{I} \models_p \varphi \mbox{ iff } \varphi \in \mathcal{I}$ for all other cases. Two formulae $\varphi$ and $\psi$ are propositionally equivalent, denoted $\varphi \sim \psi$, if $\mathcal{I} \models_p \varphi \Leftrightarrow \mathcal{I} \models_p \psi$ for all sets of formulae $\mathcal{I}$. The \textit{(propositional) equivalence class} of a formula $\varphi$ is denoted $[\varphi]_\sim$. The \textit{(propositional) quotient set} of a set of formulae $\Psi$ is denoted $\Psi_{/\sim}$.
\end{definition}
\begin{restatable}{lemma}{lemmapropequivpreserved}
\label{lemma_propequiv_preserved}
    Let $f$ be a function on formulae such that $f(\top) = \top$, $f(\bot) = \bot$, and for all formulae $\varphi$ and $\psi$, $f(\varphi \land \psi) = f(\varphi) \land f(\psi)$ and $f(\varphi \lor \psi) = f(\varphi) \lor f(\psi)$. Then, for all pairs of formulae $\varphi$ and $\psi$, if $\varphi \sim \psi$ then $f(\varphi) \sim f(\psi)$.
\end{restatable}

\section{Encoding the Past}
\label{section_encodingthepast}
Informally, given a formula $\varphi$ and word $w$, our aim is to define a function that consumes a given finite prefix of $w$, of arbitrary length $t$, and produces a new formula $\varphi'$, such that the suffix $w_t$ satisfies $\varphi'$ iff $w$ satisfies $\varphi$. This function will serve as the foundation for defining the state spaces and transition relations of the automata that we are to construct. For standard $LTL$, defining such a function is straightforward using the local semantics of $LTL$. With the introduction of past operators the situation becomes more complicated. As a prefix of $w$ is consumed we lose the information about the past therein, and must instead encode this information in the rewritten formula. The key insight is that this can be accomplished by rewriting strong past operators of $\varphi$ \textendash{} the operators $\ltlY, \ltlS$, and $\ltlNWS$ \textendash{} into their weak counterparts $\ltlWY, \ltlWS$, and $\ltlNS$, respectively, and vice versa, based on the consumed input. This section makes this idea precise.
\begin{definition}[Weakening and strengthening formulae]
The \textit{weakening} $\varphi_\mathcal{W}$ and strengthening $\varphi_\mathcal{S}$ of a formula $\varphi$ is defined by case distinction on $\varphi$ as
\begin{equation*}
\begin{aligned}  
    (\ltlY \psi)_\mathcal{W} &\coloneqq \ltlWY \psi & (\psi \ltlS \xi)_\mathcal{W} &\coloneqq \psi \ltlWS \xi & (\psi \ltlNWS \xi)_\mathcal{W} &\coloneqq \psi \ltlNS \xi \\
    (\ltlWY \psi)_\mathcal{S} &\coloneqq \ltlY \psi & (\psi \ltlWS \xi)_\mathcal{S} &\coloneqq \psi \ltlS \xi & (\psi \ltlNS \xi)_\mathcal{S} &\coloneqq \psi \ltlNWS \xi.
\end{aligned}
\end{equation*}
For all other cases we have $\varphi_\mathcal{W} \coloneqq \varphi$ and $\varphi_\mathcal{S} \coloneqq \varphi$.
\end{definition}
\begin{definition}[Rewriting past operators under sets]
Given a formula $\varphi$ and set of past formulae $C$, we write $\varphi\rw{C}$ to denote the result of weakening or strengthening the past operators in the syntax tree of $\varphi$ according to $C$ while otherwise maintaining its structure, as per the following inductive definition:
\begin{alignat*}{2}
    a\rw{C} &\coloneqq a && \mbox{($a$ atomic)} \\
    (\ltlUNOP \psi)\rw{C} &\coloneqq 
    \begin{cases}
    \begin{aligned}
        &(\ltlUNOP \psi\rw{C})_\mathcal{W} && (\ltlUNOP \psi \in C) \\
        &(\ltlUNOP  \psi\rw{C})_\mathcal{S} && \mbox{(otherwise)}
    \end{aligned}
    \end{cases}
    && \mbox{($\ltlUNOP$unary)} \\
    (\psi \ltlBINOP \xi)\rw{C} &\coloneqq 
    \begin{cases}
    \begin{aligned}
        &(\psi\rw{C} \ltlBINOP \xi\rw{C})_\mathcal{W} && (\psi \ltlBINOP \xi \in C) \\
        &(\psi\rw{C} \ltlBINOP \xi\rw{C})_\mathcal{S} && \mbox{(otherwise)}.
    \end{aligned}
    \end{cases}
    && \mbox{($\ltlUNOP$binary)}.
\end{alignat*}
We overload this definition to sets: given a set of formulae $S$, we define $S\rw{C} \coloneqq \{ s\rw{C} \mid s \in S \}$. 
\end{definition}
\begin{example}
Consider the formula $\varphi = \ltlY (p \ltlWS q)$ and set $C = \{ \varphi \}$. We then have $\varphi\rw{C} = \ltlWY (p \ltlS q)$. For the same formula $\varphi$ and set $C = \{ p \ltlWS q)$, we instead have $\varphi\rw{C} = \varphi$.
\end{example}
\begin{definition}[Weakening conditions]
Given a past formula $\varphi$, we define the weakening condition function $\wc(\varphi)$:
\begin{equation*}
\begin{aligned}
\wc(\ltlY \psi) &\coloneqq \psi & \wc(\psi \ltlS \xi) &\coloneqq \xi & \wc(\psi \ltlNWS \xi) &\coloneqq \psi \land \xi \\
\wc(\ltlWY \psi) &\coloneqq \psi & \wc(\psi \ltlWS \xi) &\coloneqq \psi \lor \xi & \wc(\psi \ltlNS \xi) &\coloneqq \xi.
\end{aligned}
\end{equation*}
\end{definition}
The weakening condition for a past formula serves as a requirement that must hold immediately in order to justify weakening the formula for the next time step. More precisely, if $w \models \wc(\varphi)$ then it is enough to check that $w_1 \models \varphi_\mathcal{W}$ to conclude that $(w, 1) \models \varphi$.
\begin{example}
    Let $\varphi = \ltlX (p \ltlS q)$ and $w$ be a word such that $w_{02} = \{ q \} \{ p \}$. By establishing that $w \models \wc(p \ltlS q) = q$, we can ``forget'' about the initial letter $\{ q \}$; it is enough to check that $w_1 \models p \ltlWS q$ to conclude that $(w, 1) \models p \ltlS q$, and hence that $w \models \varphi$. 
\end{example}
This suggests that there exists a set of past subformulae of $\varphi$ that precisely captures the information contained in the initial letter of $w$ required to evaluate all past subformulae of $\varphi$ at every point in time. This is the main result of this section, which we now summarize.
\begin{definition}[Sets of entailed subformulae]
\label{definition_set_entailed}
Let $\varphi$ be a formula, $w$ a word, and $t \in \mathbb{N}$. The set of past subformulae of $\varphi$ entailed by $w$ at $t$ is inductively defined as
\begin{equation*}
    \entconw{\varphi}{0}{w} \coloneqq \{ \psi \in \psf(\varphi) \mid \psi = \psi_\mathcal{W} \}
    \qquad
    \entconw{\varphi}{t}{w} \coloneqq \{ \psi \in \psf(\varphi\rw{C^w_{\varphi, t - 1}}) \mid w_{t-1} \models \wc(\psi) \} \quad (t > 0).
\end{equation*}
\end{definition}
When the word $w$ above is clear from the context, we simply write $\entcon{\varphi}{t}$. Given $t \in \mathbb{N}$ we denote the sequence $\entcon{\varphi}{0}, \dots, \entcon{\varphi}{t}$ of length $t + 1$ by $\entseq{\varphi}{t}$. Given a sequence $\vec{C} = C_0, C_1, \dots, C_t$ of sets of past formulae, there exists a set $C \subseteq \psf(\varphi)$ that has the same effect in a rewrite as the sequential application of rewrites of $\vec{C}$, i.e., such that $\varphi\rw{C} = \varphi\rw{C_0}\rw{C_1}\dots\rw{C_t}$. This is the set $\{ \psi \in \psf(\varphi) \mid \psi\rw{C_0}\rw{C_1}\dots\rw{C_t} = (\psi\rw{C_0}\rw{C_1}\dots\rw{C_t})_\mathcal{W} \}$, which we denote by $\circ \, \vec{C}$. In particular, there exists a set that captures the sequence of sets of entailed subformulae, denoted $\compseq{\varphi}{t}$.
\begin{restatable}{lemma}{lemmaentailedsequence}
\label{lemma_entailed_sequence}
Given a formula $\varphi$, word $w$, and $t \in \mathbb{N}$, we have $(w, t) \models \varphi$ iff $w_t \models \varphi\rw{\compseq{\varphi}{t}}$.
\end{restatable}
With these definitions in place, we are ready to define the after-function for $\mathit{pLTL}$.
\section{The after-function for \texorpdfstring{$\mathit{pLTL}$}{pLTL}}
\label{section_afterfunction}
The after-function is the foundation for defining the states and transition relations of all automata used in our $\mathit{pLTL}$-to-DRA translation. We begin by defining the \textit{local} after-function:
\begin{definition}[The local after-function]
    Given a formula $\varphi$, a letter $\sigma \in 2^{\Var(\varphi)}$, and a set of past formulae $C$, we inductively define the local after-function $\afc$ mutually with the local past update-function $\puc$ as follows:
\begin{equation*}
\begin{aligned}
\afc(\top, \sigma, C) &\coloneqq \top \\
\afc(\bot, \sigma, C) &\coloneqq \bot \\
\afc(p, \sigma, C) &\coloneqq \mbox{if } (p \in \sigma) \mbox{ then } \top \mbox{ else } \bot \\
\afc(\neg p, \sigma, C) &\coloneqq \mbox{if } (p \in \sigma) \mbox{ then } \bot \mbox{ else } \top \\
\afc(\ltlX \psi, \sigma, C) &\coloneqq \puc(\psi, \sigma, C) \\
\afc(\ltlY \psi, \sigma, C) &\coloneqq \bot \\
\afc(\ltlWY \psi, \sigma, C) &\coloneqq \top \\
\afc(\psi \ltlBINOP \xi, \sigma, C) &\coloneqq \afc(\psi, \sigma, C) \ltlBINOP \afc(\xi, \sigma, C) & (\ltlUNOP \in \{ \land, \lor \}) \\
\afc(\psi \ltlBINOP \xi, \sigma, C) &\coloneqq \afc(\xi, \sigma, C) \lor \afc(\psi, \sigma, C) \land \puc(\psi \ltlBINOP \xi, \sigma, C)
& (\ltlUNOP \in \{ \ltlUsymb, \ltlWsymb \}) \\
\afc(\psi \ltlBINOP \xi, \sigma, C) &\coloneqq \afc(\xi, \sigma, C) \land (\afc(\psi, \sigma, C) \lor \puc(\psi \ltlBINOP \xi, \sigma, C))
& (\ltlUNOP \in \{ \ltlRsymb, \ltlMsymb \}) \\
\afc(\psi \ltlBINOP \xi, \sigma, C) &\coloneqq \afc(\wc(\psi \ltlBINOP \xi), \sigma, C))
& (\ltlUNOP \in \{ \ltlSsymb, \ltlWS, \ltlNWSsymb, \ltlNSsymb \}),
\end{aligned}
\end{equation*}
where
\begin{equation*}
\puc(\varphi, \sigma, C) \coloneqq \varphi\rw{C} \land 
\bigwedge_{\psi \in \psf(\varphi) \cap C} \afc(\wc(\psi), \sigma, C).   
\end{equation*}
\end{definition}
Intuitively, in the context of reading the initial letter $\sigma$ of the word $w$, the local after-function decomposes $\varphi$ into parts that can be fully evaluated using $\sigma$ and immediately be replaced with $\top$ or $\bot$, and parts that can only be partially evaluated using $\sigma$. The resulting formula is then left to be further evaluated in the future; in the automaton, it corresponds to the state reached upon reading $\sigma$ from the state corresponding to $\varphi$. Crucially, the past subformulae of the partially evaluated part are updated by $\puc$, using the information in $C$. Here, $C$ is to be thought of as a guess of the past subformulae of $\varphi$ whose weakening conditions hold upon reading $\sigma$. Finally, the weakening conditions that must hold to justify the guess are added conjunctively. Observe that, as these weakening conditions may contain subformulae referring to the future, it may not be possible to fully evaluate them immediately; this motivates the recursive application of $\afc$ in $\puc$.
\begin{definition}[The extended local after-function]
Given a (possibly empty) finite word $w$ of length $n$ and sequence of sets of past formulae $\vec{C}$ of length $n + 1$, we extend $\afc$ to $w$ and $\vec{C}$ as follows:
\begin{equation*}
\begin{aligned}
    \afc(\varphi, \epsilon, \vec{C}_{01}) &\coloneqq \varphi &&(n = 0) \\
    \afc(\varphi, w_{0n}, \vec{C}_{0(n+1)}) &\coloneqq \afc\left(\varphi_{n - 1}, w_{(n-1)n},\vec{C}_{n(n + 1)} \right) &&(n > 0),
\end{aligned}
\end{equation*}
where $\varphi_{n} = \afc \bigl(\varphi, w_{0n}, \vec{C}_{0(n + 1)} \bigr)$. 
\end{definition}
Observe that the initial set of $\vec{C}$ in the above definition is discarded; this is to match the sequence of Definition~\ref{definition_set_entailed}. For formulae where no future operators are nested inside past operators, the set of entailed subformulae is completely determined by the prefix $w_{0n}$. This is not the case in general, however, and so the (global) after-function is defined to consider all possible subsets of past subformulae of $\varphi$ as a disjunction.
\begin{definition}[The after-function]
\label{definition_after_function}
Let $\varphi$ be a formula and $\sigma$ a letter. The after-function $\af$ is defined as:
\begin{equation*}
        \af(\varphi, \sigma) \coloneqq \bigvee_{C \in 2^{\psf(\varphi)}} 
        \afc(\varphi, \sigma, C).
\end{equation*}
The extension of $\af$ to finite words is done in the natural way: given a formula $\varphi$ and word $w$ of length $n$, we define
\begin{equation*}
\begin{aligned}    
    \af(\varphi, \epsilon) &\coloneqq \varphi &&(n = 0) \\
    \af(\varphi, w_{0n}) &\coloneqq \af(\af(\varphi, w_{0(n-1)}), w_{(n-1)n)}) &&(n > 0).
\end{aligned}
\end{equation*}
\end{definition}
\begin{example}
    Let $\varphi = \ltlX (p \ltlS \ltlX q)$. Observe that $\varphi \equiv \ltlX(p \land q \lor \ltlX q)$. Upon reading a letter $\sigma$ we can guess that the ``since'' started holding at the current point, corresponding to the set $C = \{ p \ltlS \ltlX q \}$ and formula $p \ltlWS \ltlX q \land q$. Alternatively, we may guess that the ``since'' did \textit{not} start holding upon reading $\sigma$, corresponding to the set $C = \varnothing$ and formula $p \ltlS \ltlX q$. Hence $\af(\varphi, \sigma) = p \ltlWS \ltlX q \land q \lor p \ltlS \ltlX q$. This is equivalent (at $t = 0$) to $p \land q \lor \ltlX q$, which is what must be satisfied by $w_1$, as desired.
\end{example}
The correctness of $\af$ is established by the following theorem:
\begin{restatable}{theorem}{theoremafcorrect}
\label{theorem_af_correct}
    For every formula $\varphi$, word $w$, and $t \in \mathbb{N}$ we have $w \models \varphi$ iff $w_t \models \af(\varphi, w_{0t})$.
\end{restatable}
\section{Stability and the Master Theorem}
\label{section_stability}
We consider two fragments of $\mathit{pLTL}$: $\mu$-$\mathit{pLTL}$, the set of formulae whose future operators are members of $\{ \ltlXsymb, \ltlUsymb, \ltlMsymb \}$, and $\nu$-$\mathit{pLTL}$, the set of formulae whose future operators are members of $\{ \ltlXsymb, \ltlWsymb, \ltlRsymb \}$. Given a formula $\varphi$, we define the set $\mu(\varphi)$ of subformulae of $\varphi$ whose syntax trees are rooted with $\ltlUsymb$ or $\ltlMsymb$. Similarly, we define the set $\nu(\varphi)$ of subformulae of $\varphi$ whose syntax trees are rooted with $\ltlWsymb$ or $\ltlRsymb$.

The Master Theorem for $\mathit{pLTL}$ establishes that the language of a $\mathit{pLTL}$ formula can be decomposed into a Boolean combination of simple languages. It is motivated by two ideas:
\begin{itemize}
    \item[i)] Assume that $\varphi$ is a formula and $w$ is a word such that all subformulae in $\mu(\varphi)$ that are eventually satisfied by $w$ are infinitely often satisfied by $w$, and all subformulae in $\nu(\varphi)$ that are almost always satisfied by $w$ never fail to be satisfied by $w$. In this case, we say that $w$ is a \textit{stable} word of $\varphi$, as will be properly defined shortly. Under these circumstances, a subformula of $\varphi$ of the form $\psi \ltlU \xi$ is satisfied by $w$ iff both $\psi \ltlW \xi$ and $\ltlG \ltlF (\psi \ltlU \xi)$ are. Dually, a subformula of $\varphi$ of the form $\psi \ltlW \xi$ is satisfied by $w$ iff either $\psi \ltlU \xi$ or $\ltlF \ltlG (\psi \ltlW \xi)$ are. Hence, we can partition all words over $\text{Var}(\varphi)$ into partitions of the form $P_{M, N}$, where $w \in P_{M, N}$ iff $M$ is the set of $\mu$-$\mathit{pLTL}$-subformulae of $\varphi$ satisfied infinitely often by $w$, and $N$ the set $\nu$-$\mathit{pLTL}$-subformulae of $\varphi$ that are almost always satisfied by $w$. Given two such sets $M$ and $N$, the above implies that $\varphi$ can be rewritten into a formula that belongs to either the fragment $\mu$-$\mathit{pLTL}$ or $\nu$-$\mathit{pLTL}$, as desired.
    \item[ii)] Given a formula $\varphi$ and word $w$, there exists a point in the future from which the above holds. Indeed, if we look far ahead into the future, all subformulae of $\varphi$ that are satisfied by $w$ only finitely often will have been satisfied for the last time. In particular, this is true for the $\mu$-$\mathit{pLTL}$-subformulae of $\varphi$. Similarly, there exists a point in the future at which all subformulae of $\varphi$ that are almost always satisfied by $w$ will have failed to be satisfied by $w$ for the last time. In particular, this is true for the $\nu$-$\mathit{pLTL}$-subformulae of $\varphi$.
\end{itemize}
These two notions suggest that, given a formula $\varphi$ and word $w \in P_{M, N}$, it is possible to transform $\varphi$ using the after-function until $w$ becomes stable, and then rewrite it according to either $M$ or $N$. Since these sets are unknown, we need to consider all possible combinations of such subsets, which ultimately manifests as a number of Rabin pairs exponential in the size of the formula. For more details and further examples we refer the reader to Section 5 of Esparza et al.~\cite{esparza_unified_translation}. The exposition of this section follows the similar exposition therein. However, the Master Theorem and the lemmata that imply it require considerable ``pastification''.

We now make precise the idea expressed in ii). Given a formula $\varphi$, word $w$, and $t \in \mathbb{N}$, we define the set of subformulae in $\mu(\varphi)$ that are satisfied by $w$ at least once at $t$ and the set of subformulae in $\mu(\varphi)$ that are satisfied by $w$ infinitely often at $t$. Similarly, we define the set of subformulae in $\nu(\varphi)$ that are always satisfied by $w$ at $t$ and the set of subformulae in $\nu(\varphi)$ that are almost always satisfied by $w$ at $t$:
\begin{equation*}
\begin{aligned}
    &\F{\varphi}{w}{t} \coloneqq \{ \psi \in \mu(\varphi) \mid (w, t) \models \ltlF \psi \}
    &&\GF{\varphi}{w}{t} \coloneqq \{ \psi \in \mu(\varphi) \mid (w, t) \models \ltlG \ltlF \psi \} \\
    &\G{\varphi}{w}{t} \coloneqq \{ \psi \in \nu(\varphi) \mid (w, t) \models \ltlG \psi \}
    &&\FG{\varphi}{w}{t} \coloneqq \{ \psi \in \nu(\varphi) \mid (w, t) \models \ltlF \ltlG \psi \}.
\end{aligned}
\end{equation*}
As mentioned, we are in particular interested in the point at which the two sets in each row coincide. We express this as the word being \textit{stable} at that point. 
\begin{definition}[Stable words]
    A word $w$ is $\mu$-\textit{stable} ($\nu$-\textit{stable}) with respect to a formula $\varphi$ at index $t$ if $\F{\varphi}{w}{t} = \GF{\varphi}{w}{t}$ ($\G{\varphi}{w}{t} = \FG{\varphi}{w}{t}$). If $w$ is both $\mu$-stable and $\nu$-stable with respect to $\varphi$ at index $t$, then it is \textit{stable} with respect to $\varphi$ at index $t$.
\end{definition}
\begin{restatable}{lemma}{lemmastabilitytwo}
\label{lemma_stability_two}
    Let $\varphi$ be a formula and $w$ a word. Then there exists an index $r \in \mathbb{N}$ such that $w$ is stable with respect to $\varphi$ at all indices $t \geq r$.
\end{restatable}

The following two definitions specify how to rewrite a formula according to sets $M$ and $N$, as indicated in i):
\begin{definition}
\label{definition_M}
Let $\varphi$ be a formula and $M$ a set of $\mu$-$\mathit{pLTL}$-formulae. The formula $\varphi[M]_\nu$ is inductively defined as
\begin{equation*}
\begin{aligned}
    a[M]_\nu &\coloneqq a &&(\text{a atomic}) \\
    (\ltlUNOP \psi)[M]_\nu &\coloneqq \ltlUNOP (\psi[M]_\nu) &&(\ltlUNOP \text{unary}) \\
    (\psi \ltlBINOP \xi)[M]_\nu &\coloneqq (\psi[M]_\nu) \ltlBINOP (\xi[M]_\nu) &&(\ltlUNOP \in \{ \ltlWsymb, \ltlRsymb, \ltlSsymb, \ltlWSsymb, \ltlNWSsymb, \ltlNSsymb \}) \\
    (\psi \ltlU \xi)[M]_\nu &\coloneqq
    \mathrlap{\begin{cases}
        (\psi[M]_\nu) \ltlW (\xi[M]_\nu) &(\psi \ltlU \xi \in M) \\
        \bot &(\mbox{otherwise})
    \end{cases}}
    \\
    (\psi \ltlM \xi)[M]_\nu &\coloneqq
    \mathrlap{\begin{cases}
        (\psi[M]_\nu) \ltlR (\xi[M]_\nu) &(\psi \ltlM \xi \in M) \\
        \bot &(\mbox{otherwise}),
    \end{cases}}
\end{aligned}
\end{equation*}
\end{definition}
\begin{definition}
Let $\varphi$ be a formula and $N$ a set of $\nu$-$\mathit{pLTL}$-formulae. The formula $\varphi[N]_\mu$ is inductively defined as
\begin{equation*}
\begin{aligned}
    (\psi \ltlW \xi)[N]_\mu &\coloneqq
    \begin{cases}
        \top &(\psi \ltlW \xi \in N) \\
        (\psi[N]_\mu) \ltlU (\xi[N]_\mu) &(\mbox{otherwise})
    \end{cases} \\
    (\psi \ltlR \xi)[N]_\mu &\coloneqq
    \begin{cases}
        \top &(\psi \ltlR \xi \in N) \\
        (\psi[N]_\mu) \ltlM (\xi[N]_\mu) &(\mbox{otherwise}).
    \end{cases}
\end{aligned}
\end{equation*}
The other cases are defined by recursive descent similarly to Definition~\ref{definition_M}.
\end{definition}

Notice that once a formula has been rewritten by $M$, it becomes a 
$\nu$-formula. Dually, once a formula has been rewritten by $N$, it becomes a $\mu$-formula. In terms of the hierarchy of Manna and Pnueli~\cite{hierarchy}, these are safety and guarantee formulae, respectively.
It is relatively simple to construct deterministic automata for such formulae as they do not require complicated acceptance conditions.

\begin{restatable}[The Master Theorem for $pLTL$]{theorem}{theoremmaster}
\label{theorem_master}
Let $\varphi$ be a formula and $w$ a word stable with respect to $\varphi$ at index $r$. Then $w \models \varphi$ iff there exist $M \subseteq \mu(\varphi)$ and $N \subseteq \nu(\varphi)$ such that
\begin{equation*}
\begin{aligned}
    &1) \ w_r \models \afc(\varphi, w_{0r}, \entseq{\varphi}{r})[M\rw{\compseq{\varphi}{r}}]_\nu. \\
    &2) \ \forall \psi \in M \, . \, \forall s \, . \, \exists t \geq s \, . \, w_t \models \ltlF (\psi\rw{\compseq{\varphi}{t}}[N\rw{\compseq{\varphi}{t}}]_\mu). \\
    &3) \ \forall \psi \in N \, . \, \exists t \geq 0 \, . \, w_t \models \ltlG (\psi\rw{\compseq{\varphi}{t}}[M\rw{\compseq{\varphi}{t}}]_\nu).
\end{aligned}
\end{equation*}
\end{restatable}

The statements of the Master Theorem in the only if-direction can be significantly strengthened.
The existential quantification over $t$ in premise 2 can be made universal. 
The statement $w_t \models \ldots$ in premise 3 holds for all $t' \geq r$. The strengthened version is proven in the appendix. Given the semantics of $\ltlF$ and $\ltlG$, this is not a surprise.
However, the ability to do the rewrites at every given moment is technically involved due to the incorporation of the past.
In the next section we show how to use the Master Theorem in the construction of a DRA.

\section{From \texorpdfstring{$\mathit{pLTL}$}{pLTL} to DRA}
\label{section_toautomata}

We are now ready, based on the Master Theorem, to construct a DRA for the language of a formula $\varphi$. 
We decompose the language of $\varphi$ into a Boolean combination of languages, each of which is recognized by a deterministic automaton with the relatively simple acceptance condition of B\"uchi or co-B\"uchi. 
For every possible pair of sets $M$ and $N$ of $\mu$- and $\nu$-subformulae we try to establish the premises of the Master Theorem.
Premise 1 can be checked by trying to identify a stability point $r$ from which the safety automaton for $\afc(\varphi, w_{0r}, \entseq{\varphi}{r})[M\rw{\compseq{\varphi}{t}}]_\nu$ continues forever. Whenever this safety check fails, simply try again.
Overall, this corresponds to a co-B\"uchi condition and a DCA.
Premise 2 can be checked for every $\psi\in M$ by identifying infinitely many points from which the guarantee automaton for $\ltlF (\psi\rw{\compseq{\varphi}{t}}[N\rw{\compseq{\varphi}{t}}]_\mu)$ finishes its check.
Overall, this corresponds to a B\"uchi condition and a DBA.
Premise 3 is dual to premise 2 and leads to a DCA.
The three together are combined to a DRA with one pair.
Overall, we get a DRA with exponentially many Rabin pairs; one for each choice of $M$ and $N$. We will as shorthand make use of the operators $\ltlF$ and $\ltlG$ in the construction of these automata, as described in Section~\ref{subsection_pltl}.

The major difficulty of incorporating the past into this part, is that the rewriting using the set $M$ and $N$ needs to be done with the past subformulae correctly weakened.
To facilitate this, we begin by defining an auxiliary automaton in Section~\ref{subsection_wc_automaton} that serves to track the weakening conditions that must hold in order to justify rewrites by $\placeholder\rw{\placeholder}$.

The state spaces of the automata we construct are defined in terms of the following notions. Given a formula $\varphi$ we denote by $\mathbb{B}(\varphi)$ the set of formulae $\psi$ satisfying $\sff(\psi) \subseteq \sff(\varphi) \cup \{ \top, \bot \}$. Similarly, consider a formula that is a disjunction of formulae with the property that for each such disjunct $\psi$ there exists a $C \in 2^{\psf(\varphi)}$ such that $\psi \in \mathbb{B}(\varphi\rw{C})$. We denote by $\mathbb{B}^\lor(\varphi)$ the set of all such formulae. A key observation is that the sizes of the quotient sets $\mathbb{B}(\varphi)_{/\sim}$ and $\mathbb{B}^\lor(\varphi)_{/\sim}$ are doubly exponential in the size of $\varphi$. This is proven in Appendix~\ref{appendix_reach}.

For the remainder of this section we consider a fixed formula $\varphi$ that is to be translated into a Rabin automaton and a fixed ordering $C_1, C_2, \dots, C_k$ of the elements of $2^{\psf(\varphi)}$. For simplicity, we assume $C_1 = \{ \psi \in \psf(\varphi) \mid \psi = \psi_\mathcal{W} \} = \entcon{\varphi}{0}$. We freely make use of $\af$, $\placeholder[\placeholder]_\nu$, and $\placeholder[\placeholder]_\mu$ lifted to equivalence classes of formulae. This is justified by Lemma~\ref{lemma_propequiv_preserved}.

\subsection{The Weakening Conditions Automaton}
\label{subsection_wc_automaton}
The weakening conditions automaton (WC automaton) is a bed automaton that tracks the development of weakening conditions under rewrites of the local after-function. Its states are $k$-tuples of formulae, each of which describes the requirements that remain to be verified in order to justify a sequence of rewrites. To facilitate the construction of the WC automaton, we define the following function:
\begin{restatable}{definition}{definitionrc}
Given a $k$-tuple of formulae $\psi^\times = \langle \psi_1, \psi_2, \dots, \psi_k \rangle$ and a letter $\sigma$, the \textit{rewrite condition} function is defined as $\cwc(\psi^\times, \sigma) \coloneq \langle \psi'_1, \psi'_2, \dots, \psi'_k \rangle$, where
\begin{equation*}
    \psi'_i = \bigvee_{j \in J_i} \bigg(
    \afc(\psi_j, \sigma, C_i\rw{C_j}) \land
    \bigwedge_{\xi \in C_i}
    \afc(\wc(\xi\rw{C_j}), \sigma, C_i\rw{C_j})
    \bigg),
\end{equation*}
and where $J_i \coloneqq \big\{ j \in [1..k] \mid \forall \xi, \xi' \in \psf(\varphi) \, . \, \xi\rw{C_j} = \xi'\rw{C_j} \Rightarrow \xi\rw{C_i} = \xi'\rw{C_i} \big\}$. 
\end{restatable}
The definition of $J_i$ ensures that $\psi'_i \in \mathbb{B}(\varphi\rw{C_i})$. See Appendix~\ref{appendix_reach} for a proof. \\
The rewrite condition function takes a $k$-tuple of formulae and a letter, and returns an updated $k$-tuple. In the resulting tuple, the $i^{\text{th}}$ item $\psi'_i$ is a formula that encodes the updated requirements for further applying the rewrite $\placeholder\rw{C_i}$. We remark that, for every $t > 0$, there exist indices $i \in [1..k]$ and $j \in J_i$ such that $C_i = \compseq{\varphi}{t}, C_j = \compseq{\varphi}{t - 1}$, and $C_i\rw{C_j} = \entcon{\varphi}{t}$. 

We now define the WC automaton $\mathcal{H}_\varphi \coloneqq (S, S_0, \delta_{\mathcal{H}})$ over $2^{\Var(\varphi)}$. Its set of states $S$ is $\Pi_{i = 1}^k \, \mathbb{B}(\varphi\rw{C_i})_{/\sim}$. The initial state $S_0$ is the $k$-tuple $\langle [\top]_\sim, [\bot]_\sim, \dots, [\bot]_\sim \rangle$, which represents that the set used to rewrite $\varphi$ into its initial form is known to be $C_1$. Finally, the transition relation $\delta$ is defined by $\delta_\mathcal{H}(\psi^\times, \sigma) = \cwc(\psi^\times, \sigma)$, with $\cwc$ lifted to propositional equivalence classes of formulae in the natural way.

\subsection{Verifying the Premises of the Master Theorem}
\label{subsection_verifying_master}
We now describe the automata that are capable of verifying the premises of the Master Theorem. They are all runner automata with bed automaton $\mathcal{H}_\varphi \coloneqq (S, S_0, \delta_{\mathcal{H}})$.

Given a formula $\psi \in \mu(\varphi)$ and set $N \subseteq \nu(\varphi)$ we define the runner automaton $\mathcal{B}^\psi_{2, N} \coloneqq (S, Q, Q_0, \delta, \alpha)$ over $2^{\Var(\varphi)}$ with set of states $Q \coloneqq \, \mathbb{B}^\lor(\varphi[N]_\mu \land \ltlF (\psi[N]_\mu))_{/\sim}$, initial state $Q_0 \coloneqq [\ltlF(\psi[N]_\mu)]_\sim$, and Büchi acceptance condition $\alpha \coloneqq [\top]_\sim$. Finally, its transition relation is defined as
\begin{equation*}
    \delta(\zeta, \xi^\times, \sigma) \coloneqq
    \begin{cases}
    \begin{aligned}
        &\bigvee_{i \in [1..k]} \ltlF(\psi\rw{C_i}[N\rw{C_i}]_\mu) \land \xi_i[N\rw{C_i}]_\mu &&(\zeta \sim \top) \\
        &\af(\zeta, \sigma) &&(\text{otherwise}),
    \end{aligned}
    \end{cases}
\end{equation*}
where $\xi^\times = \langle \xi_1, \xi_ 2, \dots, \xi_k \rangle$. For readability, we express $\delta$ in terms of formulae, but ask the reader to note that they represent their corresponding propositional equivalence classes.

Informally, the automaton $\mathcal{B}^\psi_{2, N}$ begins by checking that the input word $w$ satisfies the formula $\ltlF(\psi[N]_\mu)$. Because this formula is in the fragment $\mu$-$\mathit{pLTL}$, it is satisfied by $w$ iff the after-function eventually rewrites it into a propositionally true formula. At this point, the automaton restarts, checking the formula again. Because the subset of $\psf(\varphi)$ that puts $\psi[N]_\mu$ in the correct form at this point \textendash{} with respect to its past subformulae \textendash{} is unknown, all possible such sets are considered in the form of a disjunction. To each disjunct the corresponding weakening conditions, as tracked by the WC automaton, are added. It follows that $\mathcal{H_\varphi} \ltimes \mathcal{B}^\psi_{2, N}$ is able to verify premise 2) of the Master Theorem for the considered formula $\psi$ and sets $M$ and $N$.

Given a formula $\psi \in \nu(\varphi)$ and set $M \subseteq \mu(\varphi)$ we define the runner automaton $\mathcal{C}^\psi_{3, M} \coloneqq (S, Q, Q_0, \delta, \alpha)$ over $2^{\Var(\varphi)}$ with set of states $Q \coloneqq \, \mathbb{B}^\lor(\varphi[M]_\nu \land \ltlG (\psi[M]_\nu))_{/\sim}$, initial state $Q_0 \coloneqq [\ltlG(\psi[M]_\nu)]_\sim$, and co-Büchi acceptance condition $\alpha \coloneqq [\bot]_\sim$. Finally, its transition relation is defined as
\begin{equation*}
    \delta(\zeta, \xi^\times, \sigma) \coloneqq
    \begin{cases}
    \begin{aligned}
        &\bigvee_{i \in [1..k]} \ltlG(\psi\rw{C_i}[M\rw{C_i}]_\nu) \land \xi_i[M\rw{C_i}]_\nu &&(\zeta \sim \bot) \\
        &\af(\zeta, \sigma) &&(\text{otherwise}),
    \end{aligned}
    \end{cases}
\end{equation*}
where $\xi^\times = \langle \xi_1, \xi_ 2, \dots, \xi_k \rangle$. As before, the formulae of $\delta$ represent their corresponding equivalence classes. By a similar argument as before, the automaton $\mathcal{H} \ltimes \mathcal{C}^\psi_{3, M}$ is able to verify premise 3) of the Master Theorem for the formula $\psi$ and sets $M$ and $N$.

We now turn to constructing automata for verifying the first premise of the Master Theorem. Given a set $M \subseteq \mu(\varphi)$, we define the runner automaton $\mathcal{C}^1_{\varphi, M} \coloneqq (S, Q, Q_0, \delta, \alpha)$ over $2^{\Var(\varphi)}$ with set of states $Q \coloneqq \mathbb{B}^\lor(\varphi)_{/\sim} \times \mathbb{B}^\lor(\varphi[M]_\nu)_{/\sim}$, initial state $Q_0 \coloneqq \langle [\varphi]_\sim, [\varphi[M]_\nu]_\sim \rangle$, and co-Büchi acceptance condition $\alpha \coloneqq \mathbb{B}^\lor(\varphi)_{/\sim} \times \{ [\bot]_\sim \}$. Its transition relation is defined as
\begin{equation*}
    \delta(\langle \psi, \zeta \rangle, \xi^\times, \sigma) \coloneqq \\
    \begin{cases}
    \begin{aligned}
        &\Bigl\langle \af(\psi, \sigma), \bigvee_{i \in [1..k]} \af(\psi,\sigma)[M\rw{C_i}]_\nu \land \xi_i[M\rw{C_i}]_\nu \Bigr\rangle &&(\zeta \sim \bot) \\
        &\Bigl\langle \af(\psi, \sigma), \af(\zeta, \sigma) \Bigr\rangle &&(\text{otherwise}),
    \end{aligned}
    \end{cases}
\end{equation*}
where $\xi^\times = \langle \xi_1, \xi_ 2, \dots, \xi_k \rangle$. Again, we remind the reader that the formulae in the above definition represent equivalence classes of formulae. The purpose of the above automaton is to ``guess'' an index at which $w$ is stable with respect to $\varphi$, starting with the guess that it is initially stable. Both $\varphi$ and $\varphi[M]$ are evaluated in tandem. If the current guess of point of stability is incorrect, the second component will eventually collapse into a formula propositionally equivalent to $\bot$. At this point, the automaton proceeds with a new guess by reapplying $\placeholder[M]_\nu$ to $\varphi$ as it has currently been transformed by the after-function. As with the other automata, the formulae that make up the states of the WC automaton are used to justify the rewrite by each set $C_i$. The automaton $\mathcal{H}_\varphi \ltimes \mathcal{C}^1_{\varphi, M}$ is able to verify premise 1) of the Master Theorem for the formula $\varphi$ and the set $M$.

\subsection{The Rabin Automaton}

We now construct the final deterministic Rabin automaton.
The automaton is the disjunction of up to $2^n$ simpler Rabin automata; one for every possible choice of $M\subseteq \mu(\varphi)$ and $N\subseteq \mu(\varphi)$. 
Taking the disjunction of deterministic Rabin automata is possible by running automata for all disjuncts in parallel (via a product construction) and taking the acceptance condition that checks that at least one of them is accepting.
Each simpler Rabin automaton checks the three premises of the Master Theorem for its specific $M$ and $N$:
(a) $M\subseteq \GF{\varphi}{w}{0}$, (b) $N\subseteq \FG{\varphi}{w}{0}$, and 
(c) $w$ satisfies a version of $\varphi$ simplified by $M$.
Each of the three is checked by a B\"uchi or co-B\"uchi automaton. 
In order to check all three we have to consider their conjunction.
For a Rabin automaton (with one pair) it is possible to check the conjunction of B\"uchi and co-B\"uchi by running automata for all conjuncts in parallel (product construction) and using the Rabin acceptance condition (one pair) to ensure that all co-B\"uchi automata and all B\"uchi automata are accepting. 

\begin{restatable}{theorem}{theoremrabincorrectness}
\label{theorem_rabin_correctness}
Let $\varphi$ be a formula. For each $M \subseteq \mu(\varphi)$ and $N \subseteq \nu(\varphi)$, define
\begin{equation*}
\begin{gathered}
\mathcal{B}^2_{M, N} \coloneqq \bigcap_{\psi \in M} \mathcal{H}_\varphi \ltimes \mathcal{B}^\psi_{2, N} \qquad
\mathcal{C}^3_{M, N} \coloneqq \bigcap_{\psi \in N} \mathcal{H}_\varphi \ltimes \mathcal{C}^\psi_{3, M} \\
\mathcal{R}_{\varphi, M, N} \coloneqq (\mathcal{H}_\varphi \ltimes \mathcal{C}^1_{\varphi, M}) \cap  \mathcal{B}^2_{M, N} \cap \mathcal{C}^3_{M, N},
\end{gathered}
\end{equation*}
where $\mathcal{R}_{\varphi, M, N}$ has one Rabin pair. Then the following DRA over $2^{\Var(\varphi)}$ recognizes $\mathcal{\varphi}$:
\begin{equation*}
    \mathcal{A}_{DRA}(\varphi) \coloneqq \bigcup_{\substack{M \subseteq \mu(\varphi) \\ N \subseteq \nu(\varphi)}} \mathcal{R}_{\varphi, M, N}.
\end{equation*}
\end{restatable}

\begin{restatable}{corollary}{corollaryfullcomplexity}
    Let $\varphi$ be a formula of size $n + m$, where $n$ is the number of future- and propositional nodes in the syntax tree of $\varphi$, and $m$ is the number of past nodes.
    There exists a deterministic Rabin automaton recognizing $\varphi$ with doubly exponentially many states in the size of the formula and at most $2^n$ Rabin pairs.
\end{restatable}

\section{Discussion}
\label{section_discussion}
We presented a direct translation from $\mathit{pLTL}$ to deterministic Rabin automata. Starting from a formula with $n$ future subformulae and $m$ past subformulae, we produce an automaton with an optimal $2^{2^{O(n+m)}}$ states and $2^{O(n)}$ acceptance pairs. 
Our translation relies on extending the classical ``after''-function of $\mathit{LTL}$ to $\mathit{pLTL}$ by encoding memory about the past through the weakening and strengthening of embedded past operators.
We extended the Master Theorem about decomposition of languages expressed for $\mathit{LTL}$ to $\mathit{pLTL}$.

The only applicable approach (prior to our work) to obtain deterministic automata from $\mathit{pLTL}$ formulae was to convert the formula to a nondeterministic automaton \cite{piterman_handbook} and then determinize this automaton \cite{safra_determinization,piterman_determinization}.
The first can be done either directly \cite{lichtenstein_glory,piterman_handbook} or through two-way very-weak alternating automata \cite{gastin_2wvwaa}.
In any case, the first translates a formula with $n$ future operators and $m$ past operators to an automaton with $2^{O(n+m)}$ states and the second translates an automaton with $k$ states to a parity automaton with $O(k!^2)$ states and $O(k)$ priorities.
It follows that the overall complexity of this construction is $2^{2^{O((n+m) \cdot \log(n+m))}}$ states and $2^{O(n+m)}$ priorities. Our approach improves this upper bound to $2^{2^{O(m+n)}}$.
It is well known that $\mathit{pLTL}$ does not extend the expressive power of $\mathit{LTL}$.
However, conversion from $\mathit{pLTL}$ to $\mathit{LTL}$ is not viable algorithmically.
The best known translation is worst-case non-elementary~\cite{Gabbay87}, and the conversion is provably exponential \cite{markey_succinct}.
So using a conversion to $\mathit{LTL}$ as a preliminary step to determinization could result in a triple exponential construction.
We note that in the case where there are no future operators nested within past operators, it is possible to convert the past subformulae directly to deterministic automata. Then, the remaining future can be determinized independently. 
This approach has been advocated for usage of the past in reactive synthesis \cite{bloem-grone,bgrone} and implemented recently in a bespoke tool \cite{atva}.

As future work, we note that the approach of Esparza et al.~\cite{esparza_unified_translation} additionally led to translations from $\mathit{LTL}$ to nondeterministic automata, limit-deterministic automata, and deterministic automata. The same should be done for $\mathit{pLTL}$.
Their work also led to a normal form for $\mathit{LTL}$ formulae, which we believe could be generalized to work for $\mathit{pLTL}$.
The latter could have interesting relations to the temporal hierarchy of Manna and Pnueli~\cite{hierarchy}.
Such a normal form could also be related to more efficient translations from $\mathit{pLTL}$ to $\mathit{LTL}$.
Finally, this approach has led to a competitive implementation of determinization \cite{owl} and reactive synthesis \cite{strix}. Extending these implementations to handle past is of high interest.

\bibliographystyle{plainurl}
\bibliography{main}

\appendix

\section{Details on the Automata State Spaces}
\label{appendix_reach}
The purpose of this section is to show that the automata that make up our construction are well defined, in the sense that their state spaces are closed under their respective transition relations, and to establish bounds on the size of their state spaces. We begin by considering the WC automaton. Throughout this section, we assume a formula $\varphi$ and an ordering $C_1, C_2, \dots, C_k$ of the elements of $2^{\psf(\varphi)}$. Recall the definition of the rewrite condition function that underlies its transition relation:
\definitionrc*
Consider two indices $i$ and $j$. To ensure that $\psi'_i \in \mathbb{B}(\varphi\rw{C_i})$ we require that $C_i$ and $C_j$ satisfy the property that $\varphi\rw{C_j}\rw{C_i\rw{C_j}} = \varphi\rw{C_i}$. This is not the case in general, however. For example, assume that $\varphi = \ltlY p \land \ltlWY p$ and $C_j = \varnothing$ and $C_i = \{ \ltlY p \}$. Then $\varphi\rw{C_j}\rw{C_i\rw{C_j}} = (\ltlY p \land \ltlY p)\rw{C_i} = \ltlWY p \land \ltlWY p$, while $\varphi\rw{C_i} = \ltlY p \land \ltlWY p$. Hence we define the set $J_i$ based on the following notion:
\begin{definition}[Saturated sets of past formulae]
Let $i, j \leq k$. The set $C_i$ is saturated with respect to $C_j$, denoted $C_j \preceq C_i$, iff for all $\xi, \xi' \in \psf(\varphi)$ such that $\xi\rw{C_j} = \xi'\rw{C_j}$, it holds that $\xi\rw{C_i} = \xi'\rw{C_i}$.
\end{definition}
\begin{example}
As before, consider the formula $\varphi = \ltlY p \land \ltlWY p$ and the sets $C_j = \varnothing$ and $C_i = \{ \ltlY p \}$. We have that $(\ltlY p)\rw{C_j} = (\ltlWY p)\rw{C_j} = \ltlY p$. We conclude that $C_i$ is not saturated with respect to $C_j$ since $\ltlY p \in C_i$ but $\ltlWY p \notin C_i$. However, the set $\{ \ltlY p, \ltlWY p \}$ is.
\end{example}
We remark that for all $t \geq 0$ we have that $\entcon{\varphi}{t} \preceq \entcon{\varphi}{t + 1}$. Hence it suffices to consider saturated sets in the construction of the WC automaton.
\begin{lemma}
\label{lemma_saturation}
Let $j_1, j_2, \dots, j_n$ be a finite sequence of indices of length $n > 0$ such that $C_{i_j} \preceq C_{i_{j + 1}}$ for all $j < n$. Then $\varphi\rw{C_{i_1}}\rw{C_{i_2}\rw{C_{i_1}}}\dots\rw{C_{i_n}\rw{C_{i_{n - 1}}}} = \varphi\rw{C_{i_n}}$.
\begin{proof}
We prove this by induction on $n$. The base case is immediate as it amounts to proving the identity $\varphi\rw{C_{i_1}} = \varphi\rw{C_{i_1}}$. For the inductive step, let $n > 1$. By the inductive hypothesis, we have $\varphi\rw{C_{i_1}}\rw{C_{i_2}\rw{C_{i_1}}}\dots\rw{C_{i_n}\rw{C_{i_{n-1}}}} = \varphi\rw{C_{i_{n-1}}}\rw{C_{i_n}\rw{C_{i_{n-1}}}}$. Let $\psi \in \psf(\varphi)$. We will prove that (the node at the root of the syntax tree of) $\psi$ is weakened by $\placeholder\rw{C_{i_n}}$ iff it is weakened by $\placeholder\rw{C_{i_{n-1}}}\rw{C_{i_n}\rw{C_{i_{n-1}}}}$.

Assume that $\psi$ is weakened by $\placeholder\rw{C_{i_n}}$. Then $\psi \in C_{i_n}$, which implies that $\psi\rw{C_{i_{n-1}}} \in C_{i_n}\rw{C_{i_{n-1}}}$. It follows that $\psi$ is weakened by $\placeholder\rw{C_{i_{n-1}}}\rw{C_{i_n}\rw{C_{i_{n-1}}}}$. Observe that this direction does not depend on the saturation of $C_{i_n}$ with respect to $C_{i_{n-1}}$.

For the other direction, assume that $\psi$ is weakened by $\placeholder\rw{C_{i_{n-1}}}\rw{C_{i_n}\rw{C_{i_{n-1}}}}$. This means that $\psi\rw{C_{i_{n-1}}} \in C_{i_n}\rw{C_{i_{n-1}}}$. Then there exists a $\psi' \in C_{i_n}$ such that $\psi'\rw{C_{i_{n-1}}} = \psi\rw{C_{i_{n-1}}}$. By the assumption that $C_{i_{n-1}} \preceq C_{i_n}$ we know that $\psi \in C_{i_n}$, so that $\psi$ is weakened by $\placeholder\rw{C_{i_n}}$.
\end{proof}
\end{lemma}
\begin{corollary}
The quotient set $\Pi_{i=1}^k \, \mathbb{B}(\varphi\rw{C_i})_{/\sim}$ is closed under $\cwc$ lifted to $k$-tuples of propositional equivalence classes.
\begin{proof}
This follows from a simple inductive argument using Lemma~\ref{lemma_saturation} together with the fact that an application of $\afc$ results in a Boolean combination of subformulae of $\varphi$ weakened in different ways, and potentially of $\top$ and $\bot$.
\end{proof}
\end{corollary}
\begin{lemma}
\label{lemma_wc_size}
Let the size of $\varphi$ be $n + m$, where $n$ is the number of future- and propositional nodes in the syntax tree of $\varphi$, and $m$ is the number of past nodes. Then
\begin{equation*}
    \left\vert \prod_{i=1}^k \, \mathbb{B}(\varphi\rw{C_i})_{/\sim} \right\vert \leq 2^{n + 2m}.
\end{equation*}
\begin{proof}
Consider the set $\mathbb{B}(\varphi\rw{C_i})_{/\sim}$ for a given $i \in [1..k]$. Recall that this is the quotient set, under $\sim$, of formulae that are (positive) Boolean combinations of subformulae of $\varphi\rw{C_i}$ together with $\top$ and $\bot$. The definition of the propositional semantics of $\mathit{pLTL}$ of Section~\ref{section_preliminaries} implies that each equivalence class in $\mathbb{B}(\varphi\rw{C_i})_{/\sim}$ can be interpreted as a Boolean function over a set of variables corresponding to the elements of $\sff(\varphi\rw{C_i})$. Since there are $n + m$ such elements, we have,
\begin{equation*}
    \left\vert \prod_{i=1}^k \, \mathbb{B}(\varphi\rw{C_i})_{/\sim} \right\vert =
    \prod_{i=1}^{2^m} \left\vert \mathbb{B}(\varphi\rw{C_i})_{/\sim} \right\vert \leq \left(2^{2^{n + m}} \right)^{2^m} = 2^{2^{n + 2m}}.
\end{equation*}
\end{proof}
\end{lemma}
We now turn to the set $\mathbb{B}^\lor(\varphi)_{/\sim}$: the quotient set of disjunctive formulae such that there for each disjunct $\psi$ exists an index $i$ satisfying $\psi \in \mathbb{B}(\varphi\rw{C_i})$.
\begin{lemma}
The quotient set $\mathbb{B}^\lor(\varphi)_{/\sim}$ is closed under $\af$ lifted to propositional equivalence classes.
\begin{proof}
Let $\sigma$ be a letter and $[\psi]_\sim \in \mathbb{B}^\lor(\varphi)_{/\sim}$. By definition $\psi$ is propositionally equivalent to some formula $\xi_1 \lor \xi_2 \lor \dots \lor \xi_n$ such that there for each $\xi_i$ exists a $j$ satisfying $\xi_i \in \mathbb{B}(\varphi\rw{C_j})$. From the definition of the after-function the formula $\af(\psi, \sigma)$ is thus propositionally equivalent to a formula of the form $\zeta_1 \lor \zeta_2 \lor \dots \lor \zeta_m$ such that there for each $\zeta_i$ exists a $j$ and a $C \subseteq \psf(\varphi\rw{C_j})$ satisfying $\zeta_i \in \mathbb{B}(\varphi\rw{C_j}\rw{C})$. However, for each such disjunct $\zeta_i$, index $j$, and set $C$ we may choose the index $k$ such that $C_k = \{ \psi \in \psf(\varphi) \mid \psi\rw{C_j}\rw{C} = (\psi\rw{C_j}\rw{C})_\mathcal{W} \}$ to obtain $\varphi\rw{C_j}\rw{C} = \varphi\rw{C_k}$. It follows that $\af(\psi, \sigma)$ is propositionally equivalent to a formula in $\mathbb{B}^\lor(\varphi)_{/\sim}$.
\end{proof}
\end{lemma}
\begin{lemma}
\label{lemma_reach_bound}
Let the size of $\varphi$ be $n + m$, where $n$ is the number of future- and propositional nodes in the syntax tree of $\varphi$, and $m$ is the number of past nodes. Then
\begin{equation*}
    \left\vert \mathbb{B}^\lor(\varphi)_{/\sim} \right\vert \leq 2^{2^{n + 2m}}.
\end{equation*}
\begin{proof}
Since rearranging disjuncts does not affect propositional equivalence, every equivalence class of $\mathbb{B}^\lor(\varphi)_{/\sim}$ can be represented by a formula $\psi_1 \lor \psi_2 \lor \dots \lor \psi_k$, where $\psi_i \in \mathbb{B}(\varphi\rw{C_i})_{/\sim}$ for each $i \in [1..k]$. As in the proof of Lemma~\ref{lemma_wc_size}, we see that there are $2^{2^{n+m}}$ possible formulae for each $\psi_i$, up to propositional equivalence. Hence
\begin{equation*}
    \left\vert \mathbb{B}^\lor(\varphi)_{/\sim} \right\vert =
    \left\vert \prod_{i=1}^k \, \mathbb{B}(\varphi\rw{C_i})_{/\sim} \right\vert \leq 2^{2^{n + 2m}}.
\end{equation*}
\end{proof}
\end{lemma}
\begin{remark}
The bounds established in Lemma~\ref{lemma_wc_size} and Lemma~\ref{lemma_reach_bound} can easily be improved to $2^{2^{n + m}}$ by modifying the local after-function. This is done by noting that each maximal past subformula (i.e. past subformula that appears only under $\land$ or $\lor$ in the syntax tree of $\varphi$) that appears in an application of $\afc$ can immediately be evaluated and thus removed. This is because formulae rooted with one of $\ltlSsymb, \ltlWSsymb, \ltlNWSsymb$, or $\ltlNSsymb$ are initially equivalent to their weakening conditions, while formulae rooted with $\ltlY$ or $\ltlWY$ are initially equivalent to $\bot$ or $\top$, respectively.
Applying this optimization complicates notations considerably and we delay a formal treatment of it. 
\end{remark}

\section{Omitted Proofs and Related Definitions}
\subsection{Section \ref{section_preliminaries}}
\lemmapropequivpreserved*
\begin{proof}
    The proof by Esparza et al.~\cite{esparza_unified_translation} immediately generalizes to our extended definition, since all temporal subformulae are treated equally in the propositional semantics of $\mathit{pLTL}$.
\end{proof}

\subsection{Section \ref{section_encodingthepast}}

In order to prove Lemma~\ref{lemma_entailed_sequence} we state some intermediary required propositions. 
First, the formula $\varphi\rw{C}$ weakens all the past subformulae of $\varphi$ that are also in $C$. For every formula $\psi$, if $\psi_\mathcal{W}$ holds then so does $\psi_\mathcal{S}$. It follows that if $C\subseteq C'$ then weakening by $C'$  holds whenever weakening by $C$ does. 

\begin{proposition}
\label{proposition_rewrite_subset}
Let $\varphi$ be a formula, $w$ a word, and $C$ and $C'$ two sets of past formulae such that $C \subseteq C'$. Then, for arbitrary $t \in \mathbb{N}$,
\begin{equation*}
    (w, t) \models \varphi\rw{C} \Rightarrow (w, t) \models \varphi\rw{C'}.
\end{equation*}
\end{proposition}
\begin{proof}
    Using the facts that $(\sigma, t) \models \psi \Rightarrow (\sigma, t) \models \psi_\mathcal{W}$ for all subformulae $\psi$ of $\varphi$, and that $\varphi$ is in negation normal form, an inductive argument proves the proposition.
\end{proof}

We now prove that weakening indeed preserves the information that is lost by removing the first letter of the word $w$.
Namely, if the formula weakened by the correct set of formulae holds at a certain time on the shorter word (without the first letter), then the original formula holds in the same place in the full word. 

\begin{proposition}
\label{proposition_weakening_correct}
Let $\varphi$ and $\zeta$ be formulae such that $\varphi$ is a subformula of $\zeta$. Let $w$ be a word and $t \in \mathbb{N}$. Then $(w, t + 1) \models \varphi$ iff $(w_1, t) \models \varphi\rw{\entcon{\zeta}{1}}$.
\end{proposition}
\begin{proof}
We prove this by induction on the structure of the formula $\varphi$. The base cases are trivial, as they involve no past operators, as are the cases for Boolean connectives in the inductive step. We thus focus on cases in the inductive step where the top-level operator of $\varphi$ is temporal.
\begin{itemize}
    \item Case $\varphi = \psi \ltlU \xi$:
    \begin{equation*}
    \begin{aligned}
        (w, t + 1) &\models \psi \ltlU \xi
        &\Leftrightarrow &&(\text{Def. of $\models$}) \\
        \exists r > t \, . \, ((w, r) &\models \xi \mbox{ and } \forall s \in (t, r) \, . \, (w, s) \models \psi)
        &\Leftrightarrow &&(\text{IH}) \\
        \exists r \geq t \, . \, ((w_1, r) &\models \xi\rw{\entcon{\zeta}{1}} \mbox{ and } \forall s \in [t, r) \, . \, (w_1, s) \models \psi\rw{\entcon{\zeta}{1}})
        &\Leftrightarrow &&(\text{Def. of $\models$}) \\
        (w_1, t) &\models (\psi\rw{\entcon{\zeta}{1}}) \ltlU (\psi\rw{\entcon{\zeta}{1}})
        &\Leftrightarrow &&(\text{Def. of $\placeholder\rw{\placeholder}$}) \\
        (w_1, t) &\models (\psi \ltlU \xi)\rw{\entcon{\zeta}{1}}.
    \end{aligned}
    \end{equation*}

        \item Case $\varphi = \ltlY \psi:$ Assume $t = 0$. Then,
        \begin{equation*}
        \begin{aligned}
            (w, 1) &\models \ltlY \psi  &\Leftrightarrow &&(\text{Def. of $\models$}) \\
            w &\models \psi &\Leftrightarrow &&(\text{Def. of } \entcon{\zeta}{1}) \\
            \ltlY \psi &\in \entcon{\zeta}{1} &\Leftrightarrow &&(\text{Def. of} \placeholder \rw{\placeholder}) \\
            (\ltlY \psi) &\rw{\entcon{\zeta}{1}} = \ltlWY \psi &\Leftrightarrow &&(\text{Def. of $\models$}) \\
            w_1 &\models (\ltlY \psi)\rw{\entcon{\zeta}{1}}.
        \end{aligned}
        \end{equation*}
        Assume instead that $t > 0$. Then,
        \begin{equation*}
        \begin{aligned}
            (w, t + 1) &\models \ltlY \psi &\Leftrightarrow &&(\text{Def. of $\models$}) \\
            (w, t) &\models \psi &\Leftrightarrow &&(\text{IH, $t > 0$}) \\
            (w_1, t - 1) &\models \psi\rw{\entcon{\zeta}{1}} &\Leftrightarrow &&(\text{Def. of $\models$}) \\
            (w_1, t) &\models \ltlY(\psi\rw{\entcon{\zeta}{1}}) &\Leftrightarrow &&(\text{Def. of} \placeholder \rw{\placeholder}, t > 0) \\
            (w_1, t) &\models (\ltlY \psi)\rw{\entcon{\zeta}{1}}.
        \end{aligned}
        \end{equation*}
        \item Case $\varphi = \psi \ltlS \xi$: The inductive hypothesis gives us for all $s$, $r$, and $t$,
        \begin{equation*}
        \begin{aligned}
            (w, s + 1) &\models \xi \mbox{ and } \forall r \in (s + 1, t] \, . \, (w, r) \models \psi
            &\Leftrightarrow &&(\text{IH}) \\
            (w, s) &\models \xi\rw{\entcon{\zeta}{1}} \mbox{ and } \forall r \in (s, t) \, . \, (w, r) \models \psi\rw{\entcon{\zeta}{1}}.
        \end{aligned}
        \end{equation*}
        This shows that,
        \begin{equation}
        \label{since_induction}
            (w, t + 1) \models \psi \ltlS \xi \Leftrightarrow (w, t) \models (\psi\rw{\entcon{\zeta}{1}}) \ltlS (\xi\rw{\entcon{\zeta}{1}}).
        \end{equation}
        In particular, it proves the $(\Rightarrow)$-direction and both directions if $\varphi \notin \entcon{\zeta}{1}$. It remains to prove the $(\Leftarrow)$-direction given $\varphi \in \entcon{\zeta}{1}$. It suffices to assume $\xi \in \entcon{\zeta}{1}$ and $(w_1, t) \models \ltlH (\psi\rw{\entcon{\zeta}{1}})$. Then the induction hypothesis gives us $(w, r) \models \psi$, for all $r \leq t + 1$. Since $\varphi \in \entcon{\zeta}{1}$, we know that $w \models \xi$, showing that $(w, t + 1) \models \psi \ltlS \xi$.
        
        \item Case $\varphi = \ltlX \psi$:
        \begin{equation*}
        \begin{aligned}
            (w, t + 1) &\models \ltlX \psi &\Leftrightarrow &&(\text{Def. of $\models$}) \\
            (w, t + 2) &\models \psi &\Leftrightarrow &&(\text{IH}) \\
            (w_1, t + 1) &\models \psi\rw{\entcon{\zeta}{1}} &\Leftrightarrow &&(\text{Def. of} \models) \\
            (w_1, t) &\models \ltlX(\psi\rw{\entcon{\zeta}{1}}) &\Leftrightarrow &&(\text{Def. of} \placeholder \rw{\placeholder}) \\
            (w_1, t) &\models (\ltlX \psi) \rw{\entcon{\zeta}{1}}.
        \end{aligned}
        \end{equation*}
        \item Case $\varphi = \psi \ltlW \xi$:
        \begin{equation*}
        \begin{aligned}
            (w, t + 1) &\models \psi \ltlW \xi
            &\Leftrightarrow &&(\text{Def. of $\models$}) \\
            \forall r > t \, . \, ((w, r) &\models \psi \mbox{ or } \exists s \in (t, r) \, . \, (w, s) \models \xi)
            &\Leftrightarrow &&(\text{IH}) \\
            \forall r \geq t \, . \, ((w_1, r) &\models \psi\rw{\entcon{\zeta}{1}} \mbox{ or } \exists s \in [t, r) \, . \, (w_1, s) \models \xi\rw{\entcon{\zeta}{1}})
            &\Leftrightarrow &&(\text{Def. of $\models$}) \\
            (w_1, t) &\models (\psi\rw{\entcon{\zeta}{1}}) \ltlW (\xi\rw{\entcon{\zeta}{1}})
            &\Leftrightarrow &&(\text{Def. of $\placeholder\rw{\placeholder}$}) \\
            (w_1, t) &\models (\psi \ltlW \xi)\rw{\entcon{\zeta}{1}}.
        \end{aligned}
        \end{equation*}

        \item Case $\varphi = \ltlWY \psi$: The proof is identical to that of the case $\varphi = \ltlY \psi$.
        \item Case $\varphi = \psi \ltlWS \xi:$ Assume $(w, t + 1) \models \ltlH \psi$. Then $(w, t) \models \ltlH (\psi\rw{\entcon{\zeta}{1}}$, by the induction hypothesis. In particular, $w \models \psi$, so that $(\psi \ltlWS \xi) \in \entcon{\zeta}{1}$. Hence $(w_1, t) \models (\psi \ltlWS \xi)\rw{\entcon{\zeta}{1}}$.
        
        Assume $(\psi \ltlWS \xi) \in \entcon{\zeta}{1}$ and $(w_1, t) \models \ltlH (\psi\rw{\entcon{\zeta}{1}})$. The former implies that $w \models \psi \lor \xi$, while the latter together with the induction hypothesis imply that $(w, t + 1) \models \ltlH \psi$. Hence $(w, t + 1) \models \psi \ltlWS \xi$.
        The other possibilities are covered by \eqref{since_induction}.
\end{itemize}
    Since the semantics of $\varphi = \psi \ltlM \xi$, $\varphi = \psi \ltlR \xi$, $\varphi = \psi \ltlNWS \xi$, and $\varphi = \psi \ltlNS \xi$ differ only from the semantics of the already handled dual versions by the appearance of a conjunction of proper subformulae in the right operand, their proofs are omitted. 
\end{proof}

Lemma~\ref{lemma_entailed_sequence} now follows by induction: recursive weakening by the correct sets preserves the truth value of the formula under removal of a longer prefix. 

\lemmaentailedsequence*
\begin{proof}
If $t = 0$ the result is immediate since $\varphi\rw{\compseq{\varphi}{0}} = \varphi\rw{\entcon{\varphi}{0}} = \varphi$ and $w_0 = w$. If $t > 0$ then the result follows by repeated applications of Lemma~\ref{lemma_entailment_correct}.
\end{proof}

\subsection{Section \ref{section_afterfunction}}

We prove the correctness of Theorem~\ref{theorem_af_correct}. 
We state and prove intermediate results. 
The after-function tries out all possible sets of past subformulae as the sets according to which the formula is weakened.
We first show that for the disjuncts that choose the correct weakening the after-function works as expected.
That is, if at every letter read we choose the correct set of formulae to weaken the formula by, then the local after-function preserves the satisfaction of the formula. 

\begin{lemma}
\label{lemma_entailment_correct}
    Let $\varphi$ and $\zeta$ be two formulae such that $\varphi$ is a subformula of $\zeta$, $w$ be a word, and $t \in \mathbb{N}$. Then $w \models \varphi$ iff $w_t \models \afc(\varphi, w_{0t}, \entseq{\zeta}{t})$.
\end{lemma}
\begin{proof}
We prove this for $t = 1$ ($\varphi\rw{\entcon{\zeta}{1}} = \varphi\rw{\entseq{\zeta}{1}}$) by induction on the structure of $\varphi$. The full proof follows by induction on $t$. We concentrate on the inductive step for the cases where $\varphi$ is a temporal formula. By definition, $w$ satisfies the weakening conditions of all formulae in $\entcon{\zeta}{1}$. As part of the inductive hypothesis, we thus have for all subformulae $\psi$ of $\varphi$ that
\begin{equation*}
    w_1 \models \bigwedge_{\xi \in \psf(\psi) \cap  \entcon{\zeta}{1}} \afc(\wc(\xi), \sigma, \entcon{\zeta}{1}).
\end{equation*}
For this reason, we omit this component of $\puc$ in the below proof. We proceed by case distinction on $\varphi$.
\begin{itemize}
    \item Case $\varphi = \ltlX \psi$:
    \begin{equation*}
    \begin{aligned}
        w &\models \ltlX \psi &\Leftrightarrow &&(\text{Def. of } \models) \\
        (w, 1) &\models \psi &\Leftrightarrow &&(\text{Prop.}~\ref{proposition_weakening_correct}) \\
        w_1 &\models \psi\rw{\entcon{\zeta}{1}} &\Leftrightarrow &&(\text{Def. of } \afc) \\
        w_1 &\models \afc(\ltlX \psi, w_{01}, \entcon{\zeta}{1})
    \end{aligned}
    \end{equation*}
        \item Case $\varphi = \psi \ltlU \xi$: We use the identity $\psi \ltlU \xi \equiv \xi \lor \psi \land \ltlX (\psi \ltlU \xi)$. From the inductive hypothesis we get $w \models \psi \Leftrightarrow w_1 \models \afc(\psi, w_{01}, \entcon{\zeta}{1})$ and $w \models \xi \Leftrightarrow w_1 \models \afc(\xi, w_{01}, \entcon{\zeta}{1})$. It remains to show the following:
        \begin{equation*}
        \begin{aligned}
            (w, 1) &\models \psi \ltlU \xi
            &\Leftrightarrow &&(\text{Def. of} \models) \\
            \exists k \geq 1 \, . \, ((w, k) &\models \xi \mbox{ and } \forall j \in (0, k) \, . \, (w, j) \models \psi)
            &\Leftrightarrow &&(\text{Prop.~\ref{proposition_weakening_correct}}) \\
            \exists k \geq 0 \, . \, ((w_1, k) &\models \xi\rw{\entcon{\zeta}{1}} \mbox{ and } \forall j \in [0, k) \, . \, (w_1, j) \models \psi\rw{\entcon{\zeta}{1}})
            &\Leftrightarrow &&(\text{Def. of} \models) \\
            w_1 &\models (\psi\rw{\entcon{\zeta}{1}}) \ltlU (\xi\rw{\entcon{\zeta}{1}})
            &\Leftrightarrow &&(\text{Def. of} \placeholder \rw{\placeholder}) \\
            w_1 &\models (\psi \ltlU \xi)\rw{\entcon{\zeta}{1}}.
        \end{aligned}
        \end{equation*}
        \item Case $\varphi = \psi \ltlW \xi$: We use the identity $\ltlG \psi \equiv \psi \land \ltlX (\ltlG \psi)$. By the induction hypothesis we have $w \models \psi \Leftrightarrow w_1 \models \afc(\psi, w_{01}, \entcon{\zeta}{1})$. We also have
        \begin{equation*}
        \begin{aligned}
            (w, 1) &\models \ltlG \psi &\Leftrightarrow &&(\text{Prop.~\ref{proposition_weakening_correct}}) \\
            w_1 &\models \ltlG (\psi \rw{\entcon{\zeta}{1}}) &\Leftrightarrow &&(\text{Def. of} \placeholder \rw{\placeholder}) \\
            w_1 &\models (\ltlG \psi) \rw{\entcon{\zeta}{1}}.
        \end{aligned}
        \end{equation*}
        This together with the proof for $\varphi = \psi \ltlU \xi$ completes the proof.
        \item Case $\varphi = \ltlY \psi$: Immediate, as $w \not\models \ltlY \psi$.
        \item Case $\varphi = \psi \ltlS \xi$:
        \begin{equation*}
        \begin{aligned}
            w &\models \psi \ltlS \xi &\Leftrightarrow &&(\text{Def. of } \models) \\
            w &\models \xi &\Leftrightarrow &&(\text{IH}) \\
            w_1 &\models \afc(\xi, w_{01}, \entcon{\zeta}{1}) &\Leftrightarrow &&(\text{Def. of } \afc) \\
            w_1 &\models \afc(\psi \ltlS \xi, w_{01}, \entcon{\zeta}{1}).
        \end{aligned}
        \end{equation*}
        \item Case $\varphi = \psi \ltlWS \xi$:
        \begin{equation*}
        \begin{aligned}
            w &\models \psi \ltlWS \xi &\Leftrightarrow &&(\text{Def. of } \models) \\
            w &\models \psi \lor \xi &\Leftrightarrow &&(\text{IH}) \\
            w_1 &\models \afc(\psi, w_{01}, \entcon{\zeta}{1}) \lor \afc(\xi, w_{01}, \entcon{\zeta}{1}) &\Leftrightarrow &&(\text{Def. of } \afc) \\
            w_1 &\models \afc(\psi \ltlWS \xi, w_{01}, \entcon{\zeta}{1}).
        \end{aligned}
        \end{equation*}
        \end{itemize}
        As in the proof of Proposition~\ref{proposition_weakening_correct}, and for the same reason, we omit the cases dual to the ones handled above.
\end{proof}

We now show that the correct weakening at time $t$ can be found by taking a sequence of weakenings and updating the formula recursively by every letter consumed by the after-function in turn. 

\begin{lemma}
\label{lemma_entailment}
    Let $\varphi$ and $\zeta$ be two formulae such that $\varphi$ is a subformula of $\zeta$, $w$ a word, and $t \in \mathbb{N}$. Then $w_t \models \afc(\varphi, w_{0t}, \entseq{\zeta}{t})$ iff there exists a sequence of sets of past formulae $\vec{C}$ of length $t + 1$ such that $w_t \models \afc(\varphi, w_{0t}, \vec{C})$.
\begin{proof}
We prove this for $t = 1$. The full proof follows by an inductive argument.
The ``only if''-direction is immediate, since $\entcon{\zeta}{1} \in 2^{\psf(\zeta)}$. We will prove the ``if''-direction by induction on the size of $\varphi$. The base cases are immediate, as is the inductive step if $\varphi$ is a Boolean- or past formula.
    
Assume a set $C \in 2^{\psf(\zeta)}$. Given a subformula $\psi$ of $\varphi$, we have
\begin{equation*}
\begin{aligned}
    &w_1 \models \bigwedge_{\xi \in \psf(\psi) \cap C}\afc(\wc(\xi), w_{01}, C) &\Rightarrow &&(\text{IH}) \\
    &w_1 \models \bigwedge_{\xi \in \psf(\psi) \cap C} \afc(\wc(\xi), w_{01}, \entcon{\zeta}{1}) &\Rightarrow &&(\text{Lemma \ref{lemma_entailment_correct}}) \\
    &w \models \bigwedge_{\xi \in \psf(\psi) \cap C} \wc(\xi),
\end{aligned}
\end{equation*}
which in turn implies that $\psf(\psi) \cap C \subseteq \entcon{\zeta}{1}$. From this and Proposition~\ref{proposition_rewrite_subset}, we conclude that
\begin{equation}
\label{pu_implied}
    w_1 \models \puc(\psi, w_{01}, C) \Rightarrow w_1 \models \puc(\psi, w_{01}, \entcon{\zeta}{1}).
\end{equation}
The top-level structure of $\afc(\varphi, w_{01}, C)$ and $\afc(\varphi, w_{01}, \entcon{\zeta}{1})$ are identical with respect to applications of $\afc$ and $\puc$. Since these are applied to proper subformulae of $\varphi$, the inductive hypothesis together with \eqref{pu_implied} finishes the proof.
\end{proof}
\end{lemma}

Based on these two results on the local after-function and the choice of the correct weakenings we can show that the (total) after-function is indeed correct.
That is, if the formula holds initially, then applying $\af$ to it for every letter read leads to a formula that holds over the shortened suffix (and the other direction holds as well).

\theoremafcorrect*
\begin{proof}
This follows directly from the above two lemmata:
\begin{equation*}
\begin{aligned}
    w &\models \varphi &\Leftrightarrow &&(\text{Lemma~\ref{lemma_entailment_correct}}) \\
    w_t &\models \afc(\varphi, w_{0t}, \entseq{\varphi}{t}) &\Leftrightarrow &&(\text{Lemma~\ref{lemma_entailment}}) \\
    \exists \vec{C} \, . \, w_t &\models \afc(\varphi, w_{0t}, \vec{C}) &\Leftrightarrow &&(\text{Def. of} \models \text{ and } \af) \\
    w_t &\models \af(\varphi, w_{0t}).
\end{aligned}
\end{equation*}
\end{proof}

\subsection{Section \ref{section_stability}}

In order to show that a word $w$ eventually becomes stable with respect to $\varphi$ (Lemma~\ref{lemma_stability_two}), we first show that there is stabilization with respect to single formulae. 
\begin{lemma}
\label{lemma_stability_one}
    Let $\varphi$ be a formula and $w$ a word. Then
    \begin{equation*}
    \begin{aligned}
        &1) \ \exists r \, . \, \forall t \geq r \, . \, (w, t) \models \ltlF \varphi \Leftrightarrow (w, t) \models \ltlG \ltlF \varphi \\
        &2) \ \exists r \, . \, \forall t \geq r \, . \, (w, t) \models \ltlG \varphi \Leftrightarrow (w, t) \models \ltlF \ltlG \varphi.
    \end{aligned}
    \end{equation*}
\end{lemma}
\begin{proof}
    1) We prove the $(\Rightarrow)$-direction, as the other is obvious from the semantics of $\mathit{pLTL}$. Assume $w \models \ltlF \varphi$. If $w \models \ltlG \ltlF \varphi$ then we are done. If $w \not\models \ltlG \ltlF \varphi$ there exists a $r \geq 0$ such that $\forall t \geq r \, . \, (w, t) \not\models \ltlF \varphi$, which implies that $\forall t \geq r \, . \, (w, t) \not\models \ltlG \ltlF \varphi$.

    2) We prove the $(\Leftarrow)$-direction, as the other is obvious from the semantics of $\mathit{pLTL}$. Assume $w \models \ltlF \ltlG \varphi$. Then there exists an $r \geq 0$ such that $\forall t \geq r \, . \, (w, t) \models \ltlG \varphi$, which implies that $\forall t \geq r \, . \, (w, t) \models \ltlF \ltlG \varphi$.
\end{proof}

Based on stability with respect to an individual formula, we can show stability with respect to a set by taking the maximum.

\lemmastabilitytwo*
\begin{proof}
    By Lemma~\ref{lemma_stability_one}, there exist for every subformula of $\varphi$ indices $k_1$ and $k_2$ satisfying premise 1) and 2) of its statement. The desired index is the maximum of all such indices for all subformulae of $\varphi$.
\end{proof}

On the way to proving the Master Theorem (Theorem~\ref{theorem_master}), we first establish properties of the rewriting of formulae by sets of $\mu$- and $\nu$-formulae. The following lemma states when rewriting is possible and \textendash{} depending on whether the formula we rewrite by is satisfied finitely often, infinitely often, eventually always, or always \textendash{} what it means regarding the preservation of the truth value of the formula. 

\begin{lemma}
\label{lemma_MN_properties}
    Let $\varphi$ be a formula, $w$ a word, and $t \in \mathbb{N}$. Let $M$ and $N$ be sets of formulae. Then
    \begin{equation*}
    \begin{aligned}
        &1) \mbox{ If } \F{\varphi}{w}{0} \subseteq M \mbox{ and } (w, t) \models \varphi, \mbox{then } (w, t) \models \varphi[M]_\nu. \\
        &2) \mbox{ If } M \subseteq \GF{\varphi}{w}{0} \mbox{ and } (w, t) \models \varphi[M]_\nu, \mbox{then } (w, t) \models \varphi. \\
        &3) \mbox{ If } \FG{\varphi}{w}{0} \subseteq N \mbox{ and } (w, t) \models \varphi, \mbox{then } (w, t) \models \varphi[N]_\mu. \\
        &4) \mbox{ If } N \subseteq \G{\varphi}{w}{0} \mbox{ and } (w, t) \models \varphi[N]_\mu, \mbox{then } (w, t) \models \varphi.
    \end{aligned}
    \end{equation*}
\end{lemma}
\begin{proof}
We will prove this using induction on the structure of $\varphi$. We only consider the inductive step, and only a representative for the cases where $\placeholder[M]_\nu$ or $\placeholder[N]_\mu$ affect the operator at the root of the syntax tree of $\varphi$; the other cases are proven by straightforward induction and the base steps are trivial.

    1) Assume $\F{\varphi}{w}{0} \subseteq M$. Let $\varphi = \psi \ltlU \xi$ and assume $(w, t) \models \psi \ltlU \xi$. Then $\psi \ltlU \xi \in M$. We have
    \begin{equation*}
    \begin{aligned}
        (w, t) &\models \psi \ltlU \xi &\Rightarrow &&(\text{IH}) \\
        (w, t) &\models (\psi[M]_\nu) \ltlU (\xi[M]_\nu) &\Rightarrow &&(\text{Def. of} \models) \\
        (w, t) &\models (\psi[M]_\nu) \ltlW (\xi[M]_\nu) &\Rightarrow &&(\text{Def. of $M$, $\varphi \in M$}) \\
        (w, t) &\models (\psi \ltlU \xi)[M]_\nu.
    \end{aligned}
    \end{equation*}

    2) Assume $M \subseteq \GF{\varphi}{w}{0}$. Let $\varphi = \psi \ltlU \xi$ and assume $(w, t) \models (\psi \ltlU \xi)[M]_\nu$. Then $\psi \ltlU \xi \in M$, since $(w, t) \not\models \bot$. Since $M \subseteq \GF{\varphi}{w}{0}$, we have $w \models \ltlG \ltlF (\psi \ltlU \xi)$, and in particular $(w, t) \models \ltlF \xi$. Hence
    \begin{equation*}
    \begin{aligned}
        (w, t) &\models (\psi \ltlU \xi)[M]_\nu &\Rightarrow &&(\text{Def. of $M$, $\varphi \in M$}) \\
        (w, t) &\models (\psi[M]_\nu) \ltlW (\xi[M]_\nu) &\Rightarrow &&(\text{IH}) \\
        (w, t) &\models \psi \ltlW \xi &\Rightarrow &&((w, t) \models \ltlF \xi) \\
        (w, t) &\models \psi \ltlU \xi.
    \end{aligned}
    \end{equation*}

    3) Assume $\FG{\varphi}{w}{0} \subseteq N$. Let $\varphi = \psi \ltlW \xi$ and assume $(w, t) \models \psi \ltlW \xi$. If $\psi \ltlW \xi \in N$, then $(\psi \ltlW \xi)[N]_\mu = \top$, and we are done. Assume $\psi \ltlW \xi \notin N$. Since $\FG{\varphi}{w}{0} \subseteq N$, this means that $w \not\models \ltlF \ltlG (\psi \ltlW \xi)$. In particular, it means that $w \not\models \ltlF \ltlG (\psi \lor \xi)$, implying that $(w, t) \not\models \ltlG \psi$. Thus, by using the fact that $\psi \ltlW \xi \equiv \psi \ltlU \xi \lor \ltlG \psi$, we get
    \begin{equation*}
    \begin{aligned}
        (w, t) &\models \psi \ltlW \xi &\Rightarrow &&((w, t) \not\models \ltlG \psi) \\
        (w, t) &\models \psi \ltlU \xi &\Rightarrow &&(\text{IH}) \\
        (w, t) &\models (\psi[N]_\mu) \ltlU (\xi[N]_\mu) &\Rightarrow &&(\text{Def. of $N$, $\varphi \notin M$}) \\
        (w, t) &\models (\psi \ltlW \xi)[N]_\mu.
    \end{aligned}
    \end{equation*}

    4) Assume $N \subseteq \G{\varphi}{w}{0}$. Let $\varphi = \psi \ltlW \xi$ and assume $(w, t) \models (\psi \ltlW \xi)[N]_\mu$. If $\psi \ltlW \xi \in N$, then $w \models \ltlG (\psi \ltlW \xi)$, so that $(w, t) \models \ltlG (\psi \ltlW \xi)$, as required. If $\psi \ltlW \xi \notin N$, then
    \begin{equation*}
    \begin{aligned}       
        (w, t) &\models (\psi \ltlW \xi)[N]_\mu &\Rightarrow &&(\text{Def. of $N$, $\varphi \notin M$}) \\
        (w, t) &\models (\psi[N]_\mu) \ltlU (\xi[N]_\mu) &\Rightarrow &&(IH) \\
        (w, t) &\models \psi \ltlU \xi &\Rightarrow &&(\text{Def. of} \models) \\
        (w, t) &\models \psi \ltlW \xi.
    \end{aligned}
    \end{equation*}
\end{proof}

We are ready to prove the if-side of the Master Theorem.
We prove a stronger claim, which replaces existential quantifications by universal quantifications.
Namely, for the right sets of $\mu$- and $\nu$-formulae that indeed occur infinitely often and eventually always, it is possible to rewrite formulae to safety formulae and guarantee formulae and combine this with the application of weakening.

\begin{lemma}
\label{lemma_MN_=>}
    Let $\varphi$ be a formula and $w$ a word that is stable with respect to $\varphi$ at index $r$. If $M = \GF{\varphi}{w}{0}$ and $N = \FG{\varphi}{w}{0}$, then
    \begin{gather*}
        \forall \psi \in M \, . \, \forall t \geq 0 \, . \, w_t \models \ltlF (\psi\rw{\compseq{\varphi}{t}}[N\rw{\compseq{\varphi}{t}}]_\mu) \\
        \mbox{and} \\
        \forall \psi \in N \, . \, \forall t \geq r \, . \, w_t \models \ltlG (\psi\rw{\compseq{\varphi}{t}}[M\rw{\compseq{\varphi}{t}}]_\nu).
    \end{gather*}
\end{lemma}
\begin{proof}
    Let $\psi \in M$, so that $w \models \ltlG \ltlF \psi$, and let $t \geq 0$. Since $\FG{\varphi}{w}{t} = \FG{\varphi}{w}{0} \subseteq N$ we have $\FG{\varphi\rw{\compseq{\varphi}{t}}}{w_t}{0} \subseteq N\rw{\compseq{\varphi}{t}}$. By Lemma~\ref{lemma_entailed_sequence}, it holds that $w_t \models \ltlG \ltlF (\psi\rw{\compseq{\varphi}{t}})$, and we have in particular that $w_t \models \ltlF (\psi\rw{\compseq{\varphi}{t}})$. Applying Lemma~\ref{lemma_MN_properties}.3 we get $w_t \models \ltlF (\psi\rw{\compseq{\varphi}{t}}[N\rw{\compseq{\varphi}{t}}]_\mu)$.

    Let $\psi \in N$, so that $w \models \ltlF \ltlG \psi$, and let $t \geq r$. Then $(w, t) \models \ltlG \psi$, since $w$ is stable at index $t$, and so $w_t \models \ltlG (\psi\rw{\compseq{\varphi}{t}})$, by Lemma~\ref{lemma_entailed_sequence}. By assumption, we have $\F{\varphi}{w}{t} = \GF{\varphi}{w}{t} \subseteq \F{\varphi}{w}{0} \subseteq M$, which implies that $\F{\varphi\rw{\compseq{\varphi}{t}}}{w_t}{0} \subseteq M\rw{\compseq{\varphi}{t}}$. By Lemma~\ref{lemma_MN_properties}.1, we then get $w_t \models \ltlG (\psi\rw{\compseq{\varphi}{t}}[M\rw{\compseq{\varphi}{t}}]_\nu)$. 
\end{proof}

We prove the only-if-direction of the Master Theorem.
If we are able to prove that the rewrite of formulae to safety formulae and guarantee formulae combined with the application of weakening, then it is indeed the case that the original formulae is satisfied infinitely often or eventually always.

\begin{lemma}
\label{lemma_MN_<=}
    Let $\varphi$ be a formula and $w$ a word. Let $M \subseteq \mu(\varphi)$ and $N \subseteq \nu(\varphi)$. If
    \begin{gather*}
        \forall \psi \in M \, . \, \forall s \, . \, \exists t \geq s \, . \, w_t \models \ltlF (\psi\rw{\compseq{\varphi}{t}}[N\rw{\compseq{\varphi}{t}}]_\mu) \\
        \mbox{and} \\
        \forall \psi \in N \, . \, \exists t \geq 0 \, . \, w_t \models \ltlG (\psi\rw{\compseq{\varphi}{t}}[M\rw{\compseq{\varphi}{t}}]_\nu),
    \end{gather*}
    then $M \subseteq \GF{\varphi}{w}{0}$ and $N \subseteq \FG{\varphi}{w}{0}$,
\end{lemma}
\begin{proof}
Fix an enumeration $\psi_1, \psi_2, \dots, \psi_n$ of the subformulae of $M \cup N$ such that for all $i, j \in [1..n]$, if $\psi_i$ is a proper subformula of $\psi_j$, then $i < j$. We define the (unique) sequences
    \begin{equation*}
    \begin{aligned}
        M_0 &\coloneqq \varnothing & 
        M_i &\coloneqq
        \begin{cases}
            M_{i-1} \cup \{ \psi_i \} &\mbox{ if } \psi_i \in \mu(\varphi) \\
            M_{i-1} &\mbox{ if } \psi_i \in \nu(\varphi)
        \end{cases}
        \\
        N_0 &\coloneqq \varnothing & 
        N_i &\coloneqq
        \begin{cases}
            N_{i-1} &\mbox{ if } \psi_i \in \mu(\varphi) \\
            N_{i-1} \cup \{ \psi_i \} &\mbox{ if } \psi_i \in \nu(\varphi).
        \end{cases}
    \end{aligned}
    \end{equation*}
Note that $M_n = M$ and $N_n = N$, and for all $i \in [1..n]$ exactly one of $\psi_i \in M_i$ and $\psi_i \in N_i$ holds. We prove the lemma by induction on $n$. The base case with $n = 0$ is vacuously true, so we proceed by proving the inductive step with $n = i > 0$.

Assume $\psi_i \in M$. By Lemma~\ref{lemma_stability_two} and the assumption of $M$, there exists an infinite number of indices $r$ such that $\G{\varphi}{w}{r} = \FG{\varphi}{w}{r} = \FG{\varphi}{w}{0}$ and $w_r \models \ltlF (\psi_i\rw{\compseq{\varphi}{r}}[N_i\rw{\compseq{\varphi}{r}}]_\mu)$. For such an index $r$, we have by the inductive hypothesis that $\G{\varphi}{w}{r} = \FG{\varphi}{w}{0} \supseteq N_{i-1} = N_i$, so that $N_i\rw{\compseq{\varphi}{r}} \subseteq \G{\varphi\rw{\compseq{\varphi}{r}}}{w_r}{0}$. Lemma~\ref{lemma_MN_properties}.4 establishes that $w_r \models \ltlF (\psi_i\rw{\compseq{\varphi}{r}})$. By Lemma~\ref{lemma_entailed_sequence}, this means that $(w, r) \models \ltlF \psi_i$. Since the index $r$ is one of infinitely many, we have $\psi_i \in \GF{\varphi}{w}{0}$.

Assume $\psi_i \in N$. Then $w_r \models \ltlG (\psi_i\rw{\compseq{\varphi}{r}}[M_i\rw{\compseq{\varphi}{r}}]_\nu)$ for some index $r$, by assumption of $N$. By the inductive hypothesis, we have $\GF{\varphi}{w}{r} = \GF{\varphi}{w}{0} \supseteq M_{i-1} = M_i$, so that $M_i\rw{\compseq{\varphi}{r}} \subseteq \GF{\varphi\rw{\compseq{\varphi}{r}}}{w_r}{0}$. Lemma~\ref{lemma_MN_properties}.2 establishes that $w_r \models \ltlG (\psi_i\rw{\compseq{\varphi}{r}})$. By Lemma~\ref{lemma_entailed_sequence}, this means that $w \models \ltlF \ltlG \psi_i$, or $\psi_i \in \FG{\varphi}{w}{0}$.
\end{proof}

The Master Theorem follows from the two lemmata.
It complements the results about the formulae that hold infinitely often and eventually always by stating the connection between formulae holding and being able to prove that the rewrite of the main formula to a safety formula combined with the application of weakening on it.

\theoremmaster*
\begin{proof}
    $(\Rightarrow)$: Let $M = \GF{\varphi}{w}{0}$ and $N = \FG{\varphi}{w}{0}$. Properties 2) and 3) hold by Lemma~\ref{lemma_MN_=>}. By assumption $w \models \varphi$, which implies that $w_r \models \afc(\varphi, w_{0r}, \entseq{\varphi}{r})$ by Lemma~\ref{lemma_entailment_correct}. Since $\F{\varphi}{w}{r} = \GF{\varphi}{w}{r} = \GF{\varphi}{w}{0} \subseteq M$, we get $\F{\varphi\rw{\compseq{\varphi}{r}}}{w_r}{0} \subseteq M\rw{\compseq{\varphi}{r}}$. Applying Lemma~\ref{lemma_MN_properties}.1, we get $w_r \models \afc(\varphi, w_{0r}, \entseq{\varphi}{r})[M\rw{{\compseq{\varphi}{r}}}]_\nu$.

    $(\Leftarrow)$: Assume properties 1) through 3) hold. That $M \subseteq \GF{\varphi}{w}{0} = \GF{\varphi}{w}{r}$ follows from Lemma~\ref{lemma_MN_<=}, which implies that $M\rw{\compseq{\varphi}{r}} \subseteq \GF{\varphi\rw{\compseq{\varphi}{r}}}{w_r}{0}$. 
    Applying Lemma~\ref{lemma_MN_properties}.2 to the assumption 
    that $w_r \models \afc(\varphi, w_{0r}, \entseq{\varphi}{r})[M\rw{\compseq{\varphi}{r}}]_\nu$ implies $w_r \models \afc(\varphi, w_{0r}, \entseq{\varphi}{r})$. Finally, from Lemma~\ref{lemma_entailment_correct} we get $w \models \varphi$.
\end{proof}

\subsection{Section \ref{section_toautomata}}

We now prove the correctness of the construction of the DRA.

We start by stating properties of the after-function when applied to safety formulae and guarantee formulae.
Namely, a guarantee formula is eventually rewritten by the after-function to $\top$.
Dually, a safety formula is never rewritten by the after-function to $\bot$. 
This is a result of Esparza et al.~\cite{esparza_unified_translation}, which immediately extends to the case of formulae including past. 

\begin{lemma}
\label{lemma_mu_nu_limits}
    Let $\varphi$ be a formula and $w$ a word. Then
    \begin{equation*}
    \begin{aligned}
        &1) \mbox{ If } \varphi \in \mu\mbox{-}\mathit{pLTL} \mbox{ then } w \models \varphi \Leftrightarrow \exists t \, . \, \af(\varphi, w_{0t}) \sim \top. \\
        &2) \mbox{ If } \varphi \in \nu\mbox{-}\mathit{pLTL} \mbox{ then } w \models \varphi \Leftrightarrow \forall t \, . \, \af(\varphi, w_{0t}) \not\sim \bot.
    \end{aligned}
    \end{equation*}
\end{lemma}
\begin{proof}
The proof of these properties is very similar to that in Esparza et al.~\cite{esparza_unified_translation} and is omitted.
\end{proof}

The automata that we are going to construct are going to apply consecutive weakenings.
Our next steps is to delve into the workings of repeated weakenings and how they combine with the conversion of formulae to safety and guarantee formulae.
Before that, as a technicality, we extend the rewriting of conditions that are used as part of the WC automaton to finite words.

\begin{definition}[The extended rewrite condition function]
We extend the definition of $\cwc$ to finite words of length $n$ in the following way:
\begin{equation*}
\begin{aligned}
    \cwc(\psi^\times, \epsilon) &\coloneqq \psi^\times &&(n = 0) \\
    \cwc(\psi^\times, w_{0n}) &\coloneqq \cwc(\cwc(\psi^\times, w_{0(n-1)}), w_{(n-1)n)}) &&(n > 0).
\end{aligned}
\end{equation*}
\end{definition}

We turn now to the properties of weakenings.
We show that the sets of possible weakenings that we consider are rich enough to develop the correct weakening step by step.
That is, given a current weakening, there is a way to use it on top of one of a basic set of weakenings to get to the next correct weakening. 

\begin{lemma}
\label{lemma_double_context}
    Let $w$ be a word and $t \in \mathbb{N}$. Then
    \begin{equation*}
        \compseq{\varphi}{t + 1}\rw{\compseq{\varphi}{t}} = \entcon{\varphi}{t + 1}.
    \end{equation*}
\end{lemma}
\begin{proof}
We have,
\begin{equation*}
\begin{aligned}
    \compseq{\varphi}{t + 1}\rw{\compseq{\varphi}{t}} &= 
    \{ \psi \mid \psi \in \psf(\varphi) \mbox{ and } w_t \models \wc(\psi\rw{\compseq{\varphi}{t}})
    \}\rw{\compseq{\varphi}{t}} \\
    &= \{ \psi\rw{\compseq{\varphi}{t}} \mid \psi \in \psf(\varphi) \mbox{ and } w_t \models \wc(\psi\rw{\compseq{\varphi}{t}})
    \} \\
    &= \{ \psi \mid \psi \in \psf(\varphi\rw{\compseq{\varphi}{t}}) \mbox{ and } w_t \models \wc(\psi)
    \} \\
    &= \entcon{\varphi}{t + 1}.
\end{aligned}
\end{equation*}
\end{proof}

We continue with our quest for defining the correct weakening in terms of basic weakenings.
This will be required when the automata we construct blindly try out all possible weakenings (but add conditions that make sure that only correct weakenings will be used).

\begin{lemma}
\label{lemma_context_sequence_exists}
    Let $w$ be a word. Then there exists a sequence of indices $i_0, i_1, \dots$ such that $C_{i_t} = \compseq{\varphi}{t}$ for all $t \geq 0$.
\end{lemma}
\begin{proof}
    Let $i_0 = 1$, so that $C_{i_0} = \{ \psi \mid \psi \in \psf(\varphi) \mbox{ and } \psi = \psi_\mathcal{W} \} = \compseq{\varphi}{0}$. For $t \geq 0$, we inductively define $i_{t + 1}$ to be the index such that
    \begin{equation*}
        C_{i_{t + 1}} \coloneqq \{
        \psi \mid \psi \in \psf(\varphi) \mbox{ and } w_t \models \wc(\psi\rw{C_{i_t}}) 
        \}.
    \end{equation*}
    A simple inductive argument shows that $i_0, i_1, \dots$ is the sought sequence.
\end{proof}

Getting close to constructs that appear in the automata, we show that the information maintained by the WC automaton interacts well with the conversion to safety and guarantee.
Namely, we are able to take the weakening conditions maintained by the WC automaton and turn them into safety or guarantee formula.
In order to do that, we update the original $\mu$- and $\nu$-formulae that are subformulae of $\varphi$ by weakening them.
If we manage to prove that the weakened safety/guarantee weakening condition holds, then the weakenings we used are correct.
That is, we try all possible weakenings and check when applying one that its weakening condition actually works and interacts well with the conversion to safety and guarantee automata. 

\begin{lemma}
\label{lemma_rc_correct}
    Let $w = \sigma_0, \sigma_1, \dots$ be a word and $i_0, i_1, \dots$ a sequence of indices such that $C_{i_t} = \compseq{\varphi}{t}$ for all $t \in \mathbb{N}$, where $i_0 = 1$. Let $\zeta = \langle \top,\bot,\ldots, \bot\rangle$. Assume $\cwc(\zeta, w_{0t}) = \langle \psi^t_1, \psi^t_2, \dots, \psi^t_k \rangle$. Then
    \begin{equation*}
        1) \ w_t \models \psi^t_{i_t},
    \end{equation*}
    for all $t \in \mathbb{N}$. Moreover, if $M \subseteq \mu(\varphi)$, then
    \begin{equation*}
        2) \ \mbox{If } M \subseteq \GF{\varphi}{w}{0} \mbox{ then } w_t \models \psi^t_i[M\rw{C_i}]_\nu \Rightarrow C_i \subseteq \compseq{\varphi}{t},
    \end{equation*}
    for all $t$ and $i \in [1..k]$. Finally, if $N \subseteq \nu(\varphi)$ and $w$ is stable with respect to $\varphi$ at time $r$,  then
    \begin{equation*}
        3) \ \mbox {If } N \subseteq \FG{\varphi}{w}{0} \mbox{ then } w_t \models \psi^t_i[N\rw{C_i}]_\mu \Rightarrow C_i \subseteq \compseq{\varphi}{t},
    \end{equation*}
    for all $t \geq r$ and $i \in [1..k]$.
\end{lemma}
    \begin{proof}
        1) We prove this by induction on $t$. By definition we have $\psi^0_{i_0} = \top$, so that $w \models \psi^0_{i_0}$, proving the base case. For the inductive step, assume $t \geq 0$. Consider the following disjunct of $\psi^{t+1}_{i_{t+1}}$: 
        \begin{equation*}
            \afc(\psi^t_{i_t}, \sigma_t, C_{i_{t+1}}\rw{C_{i_t}}) \land \bigwedge_{\xi \in C_{i_{t+1}}} \afc(\wc(\xi\rw{C_{i_t}}), \sigma_t, C_{i_{t + 1}}\rw{C_{i_t}}).
        \end{equation*}
        By our choice of indices and Lemma~\ref{lemma_double_context} we have $C_{i_{t+1}}\rw{C_{i_t}} = \entcon{\varphi}{t+1}$. The inductive hypothesis establishes that $w_t \models \psi^t_{i_t}$ and by definition we have that $w_t \models \wc(\xi)$ for all $\xi \in \entcon{\varphi}{t+1}$. By applying Lemma~\ref{lemma_entailment_correct} we conclude that $w_{t+1}$ satisfies the above disjunct, so that $w_{t + 1} \models \psi^{t+1}_{i_{t+1}}$.

        2) If $t = 0$ then the proof is immediate. Assume $t > 0$. By assumption there exists a sequence of indices $j_0, j_1, \dots, j_t$, with $j_0 = 1$, such that
\begin{equation}
    w_t \models \bigwedge_{\xi \in C_{j_s}} f(\wc(\xi\rw{C_{j_{s-1}}}), s)[M\rw{C_{j_t}}]_\nu,
    \label{equ:clause 2 main assumption}
\end{equation}
for $s \in [1..t]$, where $f$ is inductively defined as
\begin{equation*}
\begin{aligned}
    &f(\xi, t + 1) = \xi \\
    &f(\xi, s) = f(\afc(\xi, \sigma_{s-1}, C_{j_s}\rw{C_{j_{s-1}}}), s + 1) &&(s \in [1..t]).
\end{aligned}
\end{equation*}
Let $\xi \in \sff(\varphi)$ be an arbitrary formula. We will prove that
\begin{equation*}
w_t \models f(\xi\rw{C_{j_s}}, s + 1)[M\rw{C_{j_t}}]_\nu \Rightarrow w_s \models \xi\rw{\compseq{\varphi}{s}},
\end{equation*}
for all $s \leq t$ by induction on the structure of $\xi$. 
By Equation~\ref{equ:clause 2 main assumption}, we have 
$w_t \models f(\xi\rw{C_{j_s}}, s + 1)[M\rw{C_{j_t}}]_\nu$.

The base case is trivial since propositional formulae are neither affected by rewrites of $\placeholder\rw{\placeholder}$ nor $\placeholder[\placeholder]_\nu$. 

For the inductive step, 
we assume that 
$w_t \models f(\zeta\rw{C_{j_s}}, s + 1)[M\rw{C_{j_t}}]_\nu \Rightarrow w_s \models \zeta\rw{\compseq{\varphi}{s}}$ for all $s\leq t$ and $\zeta$ proper subformula of $\xi$.
By the structure of Equation~\ref{equ:clause 2 main assumption},  we have $w_t \models f(\wc(\zeta\rw{C_{j_{t-1}}}), t)[M\rw{C_{j_{t-1}}}]_\nu$ and hence $w_{t-1} \models \wc(\zeta\rw{\compseq{\varphi}{t-1}})$.
So by the induction we have $C_{j_t} \cap \psf(\xi) \subseteq \compseq{\varphi}{t}$.

We prove that $w_t \models f(\xi\rw{C_{j_s}}, s + 1)[M\rw{C_{j_t}}]_\nu \Rightarrow w_s \models \xi\rw{\compseq{\varphi}{s}}$ by further induction on $t - s$. 
The base case is when $s = t > 0$. Then as $f(\alpha,t+1)=\alpha$ we have that the assumption
$w_t \models f(\xi\rw{C_{j_t}}, t + 1)[M\rw{C_{j_t}}]_\nu$ is rewritten to $
w_t \models \xi\rw{C_{j_t}}[M\rw{C_{j_t}}]_\nu$.
As mentioned, $C_{j_t} \cap \psf(\xi) \subseteq \compseq{\varphi}{t}$, so that $\xi\rw{\compseq{\varphi}{t}}[M\rw{\compseq{\varphi}{t}}]_\nu$ is weaker than $\xi\rw{C_{j_t}}[M\rw{C_{j_t}}]_\nu$.
Hence $
w_t \models \xi\rw{C_{j_t}}[M\rw{C_{j_t}}]_\nu$ implies 
$w_t \models \xi\rw{\compseq{\varphi}{t}}[M\rw{\compseq{\varphi}{t}}]_\nu$.
From the assumption $M \subseteq \GF{\varphi}{w}{0}$ we get $M\rw{\compseq{\varphi}{t}} \subseteq \GF{\varphi\rw{\compseq{\varphi}{t}}}{w_t}{0}$. Finally, from Lemma~\ref{lemma_MN_properties}.2 we get that $w_t \models \xi\rw{\compseq{\varphi}{t}}$, that is, $w_t \models \xi\rw{\compseq{\varphi}{s}}$.

We now turn to the inductive step with $s < t$ for some representative cases:
\begin{itemize}
\item Case $\xi = \ltlX \zeta$: We have
\begin{equation*}
\begin{aligned}
w_t &\models f((\ltlX \zeta)\rw{C_{j_s}}, s + 1)[M\rw{C_{j_t}}]_\nu \Rightarrow \\
w_t &\models f(\afc((\ltlX \zeta)\rw{C_{j_s}}, \sigma_s, C_{j_{s+1}}\rw{C_{j_s}}), s + 2)[M\rw{C_{j_t}}]_\nu \\
w_t &\models f(\puc(\zeta\rw{C_{j_s}}, \sigma_s, C_{j_{s+1}}\rw{C_{j_s}}), s + 2)[M\rw{C_{j_t}}]_\nu.
\end{aligned}
\end{equation*}
In particular we have $w_t \models f(\zeta\rw{C_{j_{s + 1}}}, s + 2)[M\rw{C_{j_t}}]_\nu$. The inductive hypothesis gives us that $w_{s + 1} \models \zeta\rw{C_{j_{s + 1}}}$. Since $C_{j_{s + 1}} \cap \psf(\zeta) \subseteq \compseq{\varphi}{s + 1}$ we get $w_{s + 1} \models \zeta\rw{\compseq{\varphi}{s + 1}}$ and thus $w_s \models (\ltlX \zeta)\rw{\compseq{\varphi}{s}}$.
\item Case $\xi = \zeta \ltlU \gamma$: We have
\begin{equation*}
\begin{aligned}
w_t &\models &&f((\zeta \ltlU \gamma)\rw{C_{j_s}}, s + 1)[M\rw{C_{j_t}}]_\nu \Rightarrow \\
w_t &\models &&f(\gamma\rw{C_{j_s}}, s + 1)[M\rw{C_{j_t}}]_\nu \lor f(\zeta\rw{C_{j_s}}, s + 1)[M\rw{C_{j_t}}]_\nu \ \land \\
&&&f(\puc((\zeta \ltlU \gamma)\rw{C_{j_s}}, \sigma_s, C_{j_{s+1}}\rw{C_{j_s}}), s + 2)[M\rw{C_{j_t}}]_\nu.
\end{aligned}
\end{equation*}
If $w_t \models f(\gamma\rw{C_{j_s}}, s + 1)[M\rw{C_{j_t}}]_\nu$ then the inductive hypothesis gives us $w_s \models \gamma\rw{\compseq{\varphi}{s}}$, so that $w_s \models (\zeta \ltlU \gamma)\rw{\compseq{\varphi}{s}}$. If instead
\begin{equation*}
w_t \models f(\zeta\rw{C_{j_s}}, s + 1)[M\rw{C_{j_t}}]_\nu \ \land \ f(\puc((\zeta \ltlU \gamma)\rw{C_{j_s}}, \sigma_s, C_{j_{s+1}}\rw{C_{j_s}}), s + 2)[M\rw{C_{j_t}}]_\nu,
\end{equation*}
then the inductive hypothesis gives us $w_s \models \zeta\rw{\compseq{\varphi}{s}}$. The satisfaction of the second conjunct implies that $w_t \models f((\zeta \ltlU \gamma)\rw{C_{j_{s + 1}}}, s + 2)[M\rw{C_{j_t}}]_\nu$. From the inductive hypothesis on $t - s$ we get that $w_{s + 1} \models (\zeta \ltlU \gamma)\rw{\compseq{\varphi}{s + 1}}$. Taking these two facts together we get $w_s \models \zeta\rw{\compseq{\varphi}{s}} \land (\ltlX(\zeta \ltlU \gamma))\rw{\compseq{\varphi}{s}}$, so that $w_s \models (\zeta \ltlU \gamma)\rw{\compseq{\varphi}{s}}$.
\item Case $\xi = \ltlY \zeta$: We have
\begin{equation*}
\begin{aligned}
w_t &\models f((\ltlY \zeta)\rw{C_{j_s}}, s + 1)[M\rw{C_{j_t}}]_\nu \Rightarrow \\
w_t &\models f(\afc((\ltlY \zeta)\rw{C_{j_s}}, \sigma_s, C_{j_{s+1}}\rw{C_{j_s}}), s + 2)[M\rw{C_{j_t}}]_\nu.
\end{aligned}
\end{equation*}
This means that $\ltlY \zeta \in C_{j_s}$, or we would have $w_t \models \bot$. Hence $\ltlY \zeta \in \compseq{\varphi}{s}$ and $w_t \models (\ltlY \zeta)\rw{\compseq{\varphi}{s}}$.
\item Case $\xi = \zeta \ltlS \gamma$: We have
\begin{equation*}
w_t \models f((\zeta \ltlS \gamma)\rw{C_{j_s}}, s + 1)[M\rw{C_{j_t}}]_\nu.
\end{equation*}
If $\zeta \ltlS \gamma \in C_{j_s}$ then either $w_t \models f(\zeta\rw{C_{j_s}}, s + 1)[M\rw{C_{j_t}}]_\nu$ or $w_t \models f(\gamma\rw{C_{j_s}}, s + 1)[M\rw{C_{j_t}}]_\nu$. We then get from the inductive hypothesis that either $w_s \models \zeta\rw{\compseq{\varphi}{s}}$ or $w_s \models \gamma\rw{\compseq{\varphi}{s}}$, and also that $\zeta \ltlS \gamma \in \compseq{\varphi}{s}$. We conclude that $w_s \models (\zeta \ltlS \gamma)\rw{\compseq{\varphi}{s}}$. The case when $\zeta \ltlS \gamma \notin C_{j_s}$ is similar.
\end{itemize}
        
3) This is completely analogous to the proof of 2) and omitted. The difference is that stability at $t$ is required in order to apply Lemma~\ref{lemma_MN_properties}.4.
\end{proof}

We are now ready to show that the B\"uchi automata we construct for a $\mu$-formula $\psi$ indeed work correctly.
That is, based on the correctness of the $\nu$-rewritings applied to it, the automaton accepts if and only if the $\mu$-formula $\psi$ holds infinitely often. 

\begin{lemma}
\label{lemma_buchi_automaton}
     Let $\varphi$ be a formula, $w$ a word, and $\psi \in \mu(\varphi)$. Let $\mathcal{H}_\varphi \coloneqq (S, S_0, \delta_{\mathcal{H}})$ be the WC automaton for $\varphi$ as defined in Section~\ref{subsection_wc_automaton}, and $\mathcal{B}^\psi_{2, N} \coloneqq (Q, Q_0, \delta, \alpha)$ the B\"uchi runner automaton as defined in Section~\ref{subsection_verifying_master}, both over $2^{\Var(\varphi)}$. Then
\begin{itemize}
    \item 1) \, If $\FG{\varphi}{w}{0} \subseteq N$ and $w$ satisfies
    \begin{equation*}
    \forall t \geq 0 \, . \, w_t \models \ltlF(\psi\rw{\compseq{\varphi}{t}}[N\rw{\compseq{\varphi}{t}}]_\mu),
    \end{equation*}
    then $\mathcal{H}_\varphi \ltimes \mathcal{B}^\psi_{2, N}$ accepts $w$.
    \item 2) \, If $N \subseteq \FG{\varphi}{w}{0}$ and there exists an accepting run of $\mathcal{H}_\varphi \ltimes \mathcal{B}^\psi_{2, N}$ on $w$ then
    \begin{equation*}
        \forall s \, . \, \exists t \geq s \, . \, w_t \models \ltlF (\psi\rw{\compseq{\varphi}{t}}[N\rw{\compseq{\varphi}{t}}]_\mu).
    \end{equation*}
\end{itemize}
\end{lemma}
\begin{proof}
1) By assumption $w \models \ltlF (\psi[N]_\mu)$. By Lemma~\ref{lemma_mu_nu_limits}, there exists a smallest index $t$ such that $\af(\ltlF(\psi)[N]_\mu, w_{0t}) \sim \top$, showing that an accepting state will be visited at least once. The subsequent state of $\mathcal{B}^\psi_{2, N}$ visited after reading $w_{t(t+1)}$ is propositionally equivalent to
\begin{equation*}
    \bigvee_{i \in [1..k]} \ltlF(\psi\rw{C_i}[N\rw{C_i}]_\mu) \land \xi^t_i[N\rw{C_i}]_\mu,
\end{equation*}
where $\xi^t_i$ is the $i^{\text{th}}$ component of the bed automaton's state. Again by assumption, we have $w_{t+1} \models \ltlF(\psi\rw{\compseq{\varphi}{t + 1}}[N\rw{\compseq{\varphi}{t + 1}}]_\mu)$. Moreover, from Lemma~\ref{lemma_rc_correct}.1 we get $w_{t+1} \models \xi^t_l$, where $l$ is such that $C_l = \compseq{\varphi}{t}$; such an $l$ exists by Lemma~\ref{lemma_context_sequence_exists}. From the assumption that $\FG{\varphi}{w}{0} \subseteq N$ we get $\FG{\varphi\rw{\compseq{\varphi}{t + 1}}}{w_{t+1}}{0} \subseteq N\rw{\compseq{\varphi}{t + 1}}$, and from an application of Lemma~\ref{lemma_MN_properties}.3 that $w_{t+1} \models \xi^t_l[N\rw{\compseq{\varphi}{t + 1}}]_\mu$, and so
\begin{equation*}
    w_{t+1} \models \bigvee_{i \in [1..k]} \ltlF(\psi\rw{C_i}[N\rw{C_i}]_\mu) \land \xi^t_i[N\rw{C_i}]_\mu.
\end{equation*}
By the same argument as before, we see that an accepting state will again be visited. The proof follows by induction.

2) Assume $u = (\xi_0,s_0), (\xi_1,s_1), \dots$ is an accepting run of $\mathcal{H}_\varphi \ltimes \mathcal{B}^\psi_{2, N}$ on $w$. Let $t$ be an arbitrary index such that $w$ is stable with respect to $\varphi$ and $s_{t-1} \in \alpha$. Then $s_t$ corresponds to the formula
\begin{equation*}
    \eta=\bigvee_{i \in [1..k]} \ltlF(\psi\rw{C_i}[N\rw{C_i}]_\mu) \land \xi^t_i[N\rw{C_i}]_\mu,
\end{equation*}
where $\xi^t_i$ is the $i^{\text{th}}$ component of $\xi_t$, as before. Let $t'$ be the minimal index such that $t'\geq t$ and $s_{t'}\in \alpha$. We then have $w_{t'} \models \af(\eta, w_{tt'})$. From Theorem~\ref{theorem_af_correct} it follows that $w_t\models \eta$. Let $i$ be such that $w_t \models \ltlF(\psi\rw{C_i}[N\rw{C_i}]_\mu \wedge \xi^t_i[N\rw{C_i}]_\mu$. In particular, $w_t\models \xi^t_i[N\rw{C_i}]_\mu$. Since $t$ was chosen such that $w$ is stable with respect to $\varphi$ at $t$ and by the assumption that $N \subseteq \FG{\varphi}{w}{0}$, we can apply Lemma~\ref{lemma_rc_correct}.3 and get $C_i \subseteq \compseq{\varphi}{t}$. Then $w_t\models \ltlF (\psi\rw{\compseq{\varphi}{t}}[N\rw{\compseq{\varphi}{t}}]_\mu)$. As $t$ was chosen as an arbitrary point such that $s_t \in \alpha$ the claim follows.
\end{proof}

We turn to the dual construction of a co-B\"uchi automaton for a $\nu$-formula $\psi$.
Based on the correctness of the $\mu$-rewritings applied to it, the automaton accepts if and only if the $\nu$-formula $\psi$ holds eventually always. 

\begin{lemma}
\label{lemma_cobuchi_automaton}
     Let $\varphi$ be a formula and $\psi \in \nu(\varphi)$. Let $\mathcal{H}_\varphi \coloneqq (S, S_0, \delta_{\mathcal{H}})$ be the WC automaton for $\varphi$ as defined in Section~\ref{subsection_wc_automaton}, and $\mathcal{C}^\psi_{3, M} \coloneqq (Q, Q_0, \delta, \alpha)$ the co-B\"uchi runner automaton as defined in Section~\ref{subsection_verifying_master}, both over $2^{\Var(\varphi)}$. Let $w$ be a word that is stable with respect to $\varphi$ at index $r$. Then
\begin{itemize}
    \item 1) \, If $\GF{\varphi}{w}{0} \subseteq M$ and $w$ satisfies
    \begin{equation*}
    \forall t \geq r \, . \, w_t \models \ltlG (\psi\rw{\compseq{\varphi}{t}}[M\rw{\compseq{\varphi}{t}}]_\nu),
    \end{equation*}
    then $\mathcal{H}_\varphi \ltimes \mathcal{C}^\psi_{3, M}$ accepts $w$.
    \item 2) \, If $M \subseteq \GF{\varphi}{w}{0}$ and there exists an accepting run of $\mathcal{H}_\varphi \ltimes \mathcal{C}^\psi_{3, M}$ on $w$ then
    \begin{equation*}
    \exists t \geq 0 \, . \, w_t \models \ltlG (\psi\rw{\compseq{\varphi}{t}}[M\rw{\compseq{\varphi}{t}}]_\nu)
\end{equation*}
\end{itemize}
\end{lemma}
\begin{proof}
    This is dual to the proof of Lemma~\ref{lemma_buchi_automaton} and omitted.
\end{proof}

The final automaton checks the derivative of the formula itself.
We show that, based on correctness of the sets $M$ and $N$, this automaton accepts exactly the words for which the formula holds. 

\begin{lemma}
\label{lemma_stability_automaton}
Let $\varphi$ be a formula and $w$ a word that is stable with respect to $\varphi$ at index $s$. Let $\mathcal{H}_\varphi \coloneqq (S, S_0, \delta_{\mathcal{H}})$ be the WC automaton for $\varphi$ as defined in Section~\ref{subsection_wc_automaton}, and $\mathcal{C}^\varphi_{1, M} \coloneqq (Q, Q_0, \delta, \alpha)$ the co-B\"uchi automaton as defined in Section~\ref{subsection_verifying_master}, both over $2^{\Var(\varphi)}$. Then,
\begin{itemize}
    \item 1) \, If $M = \GF{\varphi}{w}{0}$ and $w$ satisfies
    \begin{equation*}
    \forall r \geq s \, . \, w_r \models \afc(\varphi, w_{0r}, \entseq{\varphi}{r})[M\rw{\entseq{\varphi}{r}}]_\nu,
    \end{equation*}
    then $\mathcal{H}_\varphi \ltimes \mathcal{C}^\varphi_{1, M}$ accepts $w$.
    \item 2) \, If $M \subseteq \GF{\varphi}{w}{0}$ and there exists an accepting run of $\mathcal{H}_\varphi \ltimes \mathcal{C}^\varphi_{1, M}$ on $w$, then
    \begin{equation*}
    w \models \varphi.
    \end{equation*}
\end{itemize}
\end{lemma}
\begin{proof}
1) Assume that $M = \GF{\varphi}{w}{0}$ and that
\begin{equation*}
\forall r \geq s \, . \, w_r \models \afc(\varphi, w_{0r}, \entseq{\varphi}{r})[M\rw{\entseq{\varphi}{r}}]_\nu
\end{equation*}
Let $u = (\xi_0, \zeta_0), (\xi_1, \zeta_1), \ldots$ be the run of $\mathcal{H}_\varphi \ltimes \mathcal{C}^\varphi_{1, M}$ on $w$. 
If $u$ does not visit $\alpha$, then we are done. Otherwise, there is $t - 1 \geq s$ such that $\zeta_{t-1}\in \alpha$. Let $i_t$ be the index such that $C_{i_t} = \compseq{\varphi}{t}$; this exists by Lemma~\ref{lemma_context_sequence_exists}. By Lemma~\ref{lemma_rc_correct}.1, we have that $w_t \models \xi_{i_t}$, where $\xi_{i_t}$ is the $i_t^{\text{th}}$ updated component of the bed automaton, satisfying $C_{i_t} = \compseq{\varphi}{t}$. As $w$ is stable with respect to $\varphi$ at index $t$ and $t \geq r$, it is also stable at index $t$.
Thus $M\rw{\compseq{\varphi}{t}} = \GF{\varphi\rw{\compseq{\varphi}{t}}}{w_t}{0} = \F{\varphi\rw{\compseq{\varphi}{t}}}{w_t}{0}$. From Lemma~\ref{lemma_MN_properties}.1 we have $w_t \models \xi_{i_t}[M\rw{\compseq{\varphi}{t}}]_\nu$.

By assumption we have $w_t\models \afc(\varphi, w_{0t}, \entseq{\varphi}{t})[M\rw{\compseq{\varphi}{t}}]_\nu$. Together,
\begin{equation*}
    w_t \models \afc(\varphi, w_{0t}, \entseq{\varphi}{t})[M\rw{\compseq{\varphi}{t}}]_\nu \wedge \xi_{i_t}[M\rw{\compseq{\varphi}{t}}]_\nu,
\end{equation*}
which is a disjunct in the second component of $s_t$. By Lemma~\ref{lemma_mu_nu_limits}, this component never becomes propositionally equivalent to $\bot$ and so $u$ is accepting.

2) Let $w$ be a word accepted by $\mathcal{H}_\varphi \ltimes \mathcal{C}^\varphi_{1, M}$ and $u = (\xi_0, \zeta_0), (\xi_1, \zeta_1), \dots$ an accepting run of $\mathcal{H}_\varphi \ltimes \mathcal{C}^\varphi_{1, M}$ on $w$. 
Let $j-1$ be the final index such that the second component of $\zeta_{j-1}$ is propositionally equivalent to $\bot$. By construction the first component of $\zeta_{j-1}$ is $\af(\varphi,w_{0(j-1)})$. Then the second component of $\zeta_j$ is
\begin{equation*}
    \bigvee_{i \in [1..k]} \af(\varphi,w_{0j})[M\rw{C_i}]_\nu \land \xi^j_i[M\rw{C_i}]_\nu,
\end{equation*}
where $\xi^j_i$ is the $i^{\text{th}}$ updated component of the bed automaton's state $\xi_j$.  As the run is accepting, it follows that for some $C_i$ we have $w_j \models \af(\varphi, w_{0j})[M\rw{C_i}]_\nu$ and $w_j \models \xi^j_i[M\rw{C_i}]_\nu$. By assumption $M \subseteq \GF{\varphi}{w}{0}$ meaning that Lemma~\ref{lemma_rc_correct}.2 can be applied to yield that $C_i \subseteq \compseq{\varphi}{j}$, and so $w_j \models \af(\varphi, w_{0j})[M\rw{\compseq{\varphi}{j}}]$. Moreover, the assumption $M \subseteq \GF{\varphi}{w}{0}$ also gives us that $M\rw{\compseq{\varphi}{j}} \subseteq \GF{\varphi\rw{\compseq{\varphi}{j}}}{w_j}{0}$. Hence $w_j \models \af(\varphi, w_{0j})$ by Lemma~\ref{lemma_MN_properties}.2 and by Theorem~\ref{theorem_af_correct} we get $w \models \varphi$.
\end{proof}

We are finally ready to state the correctness of the full construction.
We show that the different assumptions made by the different automata, when checked together indeed lead to the correct conclusion.
Mostly, the different assumptions made by the different automata on each other are not really circular.

\theoremrabincorrectness*
\begin{proof}
Assume $w$ is a word such that $w \models \varphi$, where $w$ is stable with respect to $\varphi$ at index $s$. Then $w \models \afc(\varphi, w_{0r}, \entseq{\varphi}{r})$ for all $r \geq 0$ by Lemma~\ref{lemma_entailment_correct}. By Theorem~\ref{theorem_master} there exist $M \subseteq \mu(\varphi)$ and $N \subseteq \nu(\varphi)$ such that its premises hold. In particular, for $M = \GF{\varphi}{w}{0}$ and $N = \FG{\varphi}{w}{0}$, the stronger properties below hold due to Lemma~\ref{lemma_MN_properties}.1 and Lemma~\ref{lemma_MN_=>}:
\begin{equation*}
\begin{aligned}
    &1) \ \forall r \geq s \, . \, w_r \models \afc(\varphi, w_{0r}, \entseq{\varphi}{r})[M\rw{\compseq{\varphi}{r}}]_\nu. \\
    &2) \ \forall \psi \in M \, . \, \forall t \geq 0 \, . \, w_t \models \ltlF (\psi\rw{\compseq{\varphi}{t}}[N\rw{\compseq{\varphi}{t}}]_\mu). \\
    &3) \ \forall \psi \in N \, . \, \forall t \geq s \, . \, w_t \models \ltlG (\psi\rw{\compseq{\varphi}{t}}[M\rw{\compseq{\varphi}{t}}]_\nu).
\end{aligned}
\end{equation*}
This means that the premises of Lemma~\ref{lemma_stability_automaton}.1 hold, as well as the premises of Lemma~\ref{lemma_buchi_automaton}.1 for each $\psi \in M$ and of Lemma~\ref{lemma_cobuchi_automaton}.1 for each $\psi \in N$. Hence $w$ is accepted by $\mathcal{H}_\varphi \ltimes \mathcal{C}^1_{\varphi, M}$, by $\mathcal{H}_\varphi \ltimes \mathcal{B}^\psi_{2, N}$ for every $\psi \in M$, and by $\mathcal{H}_\varphi \ltimes \mathcal{C}^\psi_{3, N}$ for every $\psi \in N$. It follows that $w$ is accepted by $\mathcal{R}_{\varphi, M, N}$ for this particular choice of $M$ and $N$, and hence also by $\mathcal{A}_{DRA}(\varphi)$.

Now assume that $w$ is a word accepted by $\mathcal{A}_{DRA}(\varphi)$. This means that there exists an accepting run on $\mathcal{R}_{\varphi, M, N}$ of $w$ for some $M \subseteq \mu(\varphi)$ and $N \subseteq \nu(\varphi)$. Then $w$ is also accepted by $\mathcal{H}_\varphi \ltimes \mathcal{C}^1_{\varphi, M}$, by $\mathcal{H}_\varphi \ltimes \mathcal{B}^\psi_{2, N}$ for every $\psi \in M$, and by $\mathcal{H}_\varphi \ltimes \mathcal{C}^\psi_{3, N}$ for every $\psi \in N$. It remains to show that $M \subseteq \GF{\varphi}{w}{0}$, so that the premises of Lemma~\ref{lemma_stability_automaton}.2 hold, implying that $w \models \varphi$. We prove the stronger statement that both $M \subseteq \GF{\varphi}{w}{0}$ and $N \subseteq \FG{\varphi}{w}{0}$ hold.

The argument is very similar to that used in the proof of Lemma~\ref{lemma_MN_<=}. We again fix an enumeration $\psi_1, \psi_2, \dots, \psi_n$ of the subformulae of $M \cup N$ such that for all $i, j \in [1..n]$, if $\psi_i$ is a proper subformula of $\psi_j$, then $i < j$. We define
\begin{equation*}
\begin{aligned}
    M_0 &\coloneqq \varnothing & 
    M_i &\coloneqq
    \begin{cases}
        M_{i-1} \cup \{ \psi_i \} &\mbox{ if } \psi_i \in \mu(\varphi) \\
        M_{i-1} &\mbox{ if } \psi_i \in \nu(\varphi)
    \end{cases}
    \\
    N_0 &\coloneqq \varnothing & 
    N_i &\coloneqq
    \begin{cases}
        N_{i-1} &\mbox{ if } \psi_i \in \mu(\varphi) \\
        N_{i-1} \cup \{ \psi_i \} &\mbox{ if } \psi_i \in \nu(\varphi),
    \end{cases}
\end{aligned}
\end{equation*}
and proceed by induction on $n$. The base case when $n = 0$ is vacuously true. We proceed with the inductive step with $n = i > 0$. Assume $\psi_i \in M$. By the inductive hypothesis, $\FG{\varphi}{w}{0} \supseteq N_{i - 1} = N_i$. Thus the premises of Lemma~\ref{lemma_buchi_automaton}.2 are satisfied, so that
\begin{equation*}
    \forall r \, . \, \exists t \geq r \, . \, w_t \models \ltlF (\psi_i\rw{\compseq{\varphi}{t}}[N_i\rw{\compseq{\varphi}{t}}]_\mu).
\end{equation*}
This implies that there are infinitely many points $t$ such that $w_t \models \ltlF (\psi_i\rw{\compseq{\varphi}{t}}[N_i\rw{\compseq{\varphi}{t}}]_\mu)$. In particular, for each such $t \geq r$, since $N_i\rw{\compseq{\varphi}{t}} \subseteq \FG{\varphi\rw{\compseq{\varphi}{t}}}{w_t}{0} = \G{\varphi\rw{\compseq{\varphi}{t}}}{w_t}{0}$, we get from Lemma~\ref{lemma_MN_properties}.4 that $w_t \models \ltlF (\psi_i\rw{\compseq{\varphi}{t}})$. It follows that $\psi_i \in \GF{\varphi}{w}{0}$, which together with the inductive hypothesis shows that $M_i \subseteq \GF{\varphi}{w}{0}$. In the same manner, assuming that $\psi_i \in N$, we get
\begin{equation*}
    \exists t \geq 0 \, . \, w_t \models \ltlG (\psi_i\rw{\compseq{\varphi}{t}}[M_i\rw{\compseq{\varphi}{t}}]_\nu).
\end{equation*}
For such an index $t$, we get by Lemma~\ref{lemma_MN_properties}.2 and the same reasoning as above that $w_t \models \ltlG(\psi_i\rw{\compseq{\varphi}{t}})$. Then $\psi_i \in \FG{\varphi}{w}{0}$, which together with the inductive hypothesis establishes that $N_i \subseteq \FG{\varphi}{w}{0}$.
\end{proof}

We now analyze the size of the resulting automaton.

\corollaryfullcomplexity*
\begin{proof}
By Theorem~\ref{theorem_rabin_correctness}, the automaton $\mathcal{A}_{DRA}(\varphi)$ recognizes $\varphi$. That it has at most $2^n$ Rabin pairs is clear, since each component $\mathcal{R}_{\varphi, M, N}$ contributes with one Rabin pair and there are at most $2^n$ choices of $M$ and $N$. We will now establish bounds on the size of the state spaces for the components using the results of Lemma~\ref{lemma_wc_size} and Lemma~\ref{lemma_reach_bound} from Appendix~\ref{appendix_reach}.

Since the WC automaton is the bed automaton for $\mathcal{C}^1_{\varphi, M}$ as well as for the automata $\mathcal{B}^\psi_{2, N}$ for every $\psi \in M$ and $\mathcal{C}^\psi_{3, M}$ for every $\psi \in N$, only a single copy of it is required. As such, we will restrict our attention to the sets $Q$ of the runner automata below.

The automaton $\mathcal{B}^2_{M, N}$ is a product of Büchi automata, each of which has at most $2^{2^{n + 2m + 1}}$ states. Thus this component consists of at most $n(2^{2^{n + 2m + 1}})$ states. Observe that it is not necessary to take the product of the state spaces of the automata $\mathcal{B}^\psi_{2, N}$; they verify that simplified subformulae of $\varphi$ hold infinitely often, and so may be chained in a round-robin manner. The automaton $\mathcal{C}^3_{M, N}$, being a product of co-Büchi automata, consists of at most $(2^{2^{n + 2m + 1}})^n$ states. Finally, the automaton $\mathcal{C}^1_{\varphi, M}$ consists of at most $(2^{2^{n + 2m}})^2 = 2^{2^{n + 2m + 1}}$ states. 

For each choice of $M$ and $N$, ignoring the contribution of the WC automaton, the number of states of $\mathcal{R}_{\varphi, M, N}$ is thus bounded by
\begin{equation*}
\begin{aligned}
    n(2^{2^{n + 2m + 1}}) \cdot (2^{2^{n + 2m + 1}})^n \cdot 2^{2^{n + 2m + 1}} &= 2^{\log_2(n) + 2^{n + 2m + 1}} \cdot 2^{n 2^{n + 2m + 1}} \cdot 2^{2^{n +2m + 1}} \\
    &\leq \left( 2^{2^{2n + 2m + 1}} \right)^3 \\
    &\leq 2^{2^{2n + 2m + 3}}.
\end{aligned}
\end{equation*}

Together, with the WC automaton's contribution of at most $2^{2^{n + 2m}}$ states appearing only once and with at most $2^n$ combinations of sets $M$ and $N$, we get the following bound on the number of states of $\mathcal{A}_{DRA}$:
\begin{equation*}
\begin{aligned}
    2^{2^{n + 2m}} \cdot \left(2^{2^{2n + 2m + 3}}\right)^{2^n} =
    2^{2^{n + 2m}} \cdot 2^{2^{3n + 2m + 3}} \leq 2^{2^{3n + 2m + 4}}.
\end{aligned}
\end{equation*}
It follows that the number of states of the automaton is in $2^{2^{\mathcal{O}(\vert \varphi \vert)}}$.
\end{proof}
\end{document}